\documentclass[onefignum,onetabnum,onethmnum,onealgnum]{siamart220329}

\usepackage{amsfonts}
\usepackage{graphicx}
\usepackage{epstopdf}
\usepackage{algorithmic}
\usepackage{amsopn}
\usepackage{mathrsfs}
\usepackage{enumitem}
\usepackage{makecell}
\usepackage{colortbl}
\usepackage{booktabs}
\usepackage{xcolor}
\usepackage{longtable}
\usepackage{pdflscape}
\usepackage{multirow}
\usepackage{tikz}
\usepackage{dsfont}
\usepackage{makecell}
\usetikzlibrary{decorations.pathreplacing}

\ifpdf
  \DeclareGraphicsExtensions{.eps,.pdf,.png,.jpg}
\else
  \DeclareGraphicsExtensions{.eps}
\fi

\graphicspath{{pics/}}

\newsiamremark{remark}{Remark}
\newsiamremark{hypothesis}{Hypothesis}
\newsiamremark{condition}{Condition}
\newsiamremark{assumption}{Assumption}
\crefname{hypothesis}{Hypothesis}{Hypotheses}
\crefname{condition}{Condition}{Conditions}
\crefname{assumption}{Assumption}{Assumptions}
\newsiamthm{claim}{Claim}
\newlist{condenum}{enumerate}{1}
\setlist[condenum]{label=\alph*), ref=\thecondition~(\alph*)}
\crefalias{condenumi}{condition}

\DeclareMathOperator{\diam}{diam}
\DeclareMathOperator{\supp}{supp}

\DeclareMathOperator{\law}{Law}

\DeclareMathOperator*{\argmin}{arg\,min}
\DeclareMathOperator*{\cov}{Cov}
\DeclareMathOperator*{\lip}{Lip}

\tikzset{>=stealth}

\headers{Sampling via F\"ollmer Flow}{Z.Ding,Y.Jiao,X.Lu,Z.Yang,C.Yuan}

\title{Sampling via F\"ollmer Flow\thanks{Submitted to the editors \today.
}}

\author{
Zhao Ding\thanks{School of Mathematics and Statistics, Wuhan University, Wuhan, People’s Republic of China(\email{zd1998@whu.edu.cn}).}
\and Yuling Jiao\thanks{School of Mathematics and Statistics and Hubei Key Laboratory of Computational Science, Wuhan University, Wuhan, People’s Republic of China(\email{yulingjiaomath@whu.edu.cn}).}
\and Xiliang Lu\thanks{School of Mathematics and Statistics and Hubei Key Laboratory of Computational Science, Wuhan University, Wuhan, People’s Republic of China(\email{xllv.math@whu.edu.cn}).}
\and Zhijian Yang\thanks{School of Mathematics and Statistics and Hubei Key Laboratory of Computational Science, Wuhan University, Wuhan, People’s Republic of China(\email{zjyang.math@whu.edu.cn}).}
\and Cheng Yuan\thanks{School of Mathematics and Statistics, Wuhan University, Wuhan, People’s Republic of China(\email{yuancheng@whu.edu.cn}).}
}

\makeatletter
\newcommand*{\addFileDependency}[1]{%
  \typeout{(#1)}%
  \@addtofilelist{#1}%
  \IfFileExists{#1}{}{\typeout{No file #1.}}%
}
\makeatother

\ifpdf
\hypersetup{
  pdftitle={Sampling via F\"{o}llmer flow},
  pdfauthor={Zhao Ding, Yuling Jiao, Xiliang Lu, Zhijian Yang, Cheng Yuan}
}
\fi

\begin{document}

\maketitle

\begin{abstract}
 We introduce a novel unit-time ordinary differential equation (ODE) flow called the preconditioned F\"{o}llmer flow, which efficiently transforms a Gaussian measure into a desired target measure at time 1. To discretize the flow, we apply Euler's method, where the velocity field is calculated either analytically or through Monte Carlo approximation using Gaussian samples.
  Under reasonable conditions, we derive a non-asymptotic error bound in the Wasserstein distance between the sampling distribution and the target distribution. Through numerical experiments on mixture distributions in 1D, 2D, and  high-dimensional spaces, we demonstrate that the samples generated by our proposed flow exhibit higher quality compared to those obtained by several existing methods.
  Furthermore, we propose leveraging the F\"{o}llmer flow as a warmstart strategy for existing Markov Chain Monte Carlo (MCMC) methods, aiming to mitigate mode collapses and enhance their performance.
  Finally, thanks to the deterministic nature of the F\"{o}llmer flow, we can leverage deep neural networks to fit  the trajectory of sample evaluations. This allows us to obtain a generator for one-step sampling as a result.
\end{abstract}

\begin{keywords}
  Sampling, ODE flow, preconditioning, Monte Carlo method,  Wasserstein-2 bound
\end{keywords}

\begin{AMS}
62D05,58J65,60J60
\end{AMS}

\section{Introduction}
 Sampling from probability distributions constitutes a foundational task within the realms of statistics and machine learning \cite{cui2016109,salakhutdinov2015learning}. 
 For instance, the efficacy of Bayesian inference heavily relies on the capacity to generate samples from a posterior distribution, even when only an unnormalized density function is accessible \cite{gelman1995bayesian, changye2020markov}. Concurrently, contemporary advancements in generative learning are centered around the technique of sampling from a distribution for which the probability density remains unknown, but random samples can be obtained \cite{song2020score}.
  
To sample from high-dimensional probability distributions, a vast array of sampling techniques has been extensively developed in the literature. 
Among these, a prominent category comprises Markov Chain Monte Carlo (MCMC) methods, which encompass a variety of approaches. 
Notable examples include the Metropolis-Hastings algorithm \cite{metropolis1953equation, hastings1970monte, tierney1994markov}, the Gibbs sampler \cite{geman1984stochastic, gelfand1990sampling}, the Langevin algorithm \cite{roberts1996exponential, dalalyan2017theoretical, durmus2017nonasymptotic}, and the Hamiltonian Monte Carlo algorithm \cite{duane1987hybrid, neal2011mcmc}, among others. 
These stochastic methods give rise to an ergodic Markov chain, characterized by an invariant distribution that converges to the desired target distribution.
Under the assumption of strongly convex potentials, MCMC samplers exhibit favorable convergence properties, as established in previous works \cite{durmus2019high, Durmus2016SamplingFA, dalalyan2017further, cheng2018convergence, dalalyan2019user}. 
Moreover, researchers have explored alternative conditions to replace the strongly convex potential assumption, such as the dissipativity condition on the drift term \cite{raginsky2017non, mou2022improved, zhang2023nonasymptotic}, the local convexity condition on the potential function \cite{durmus2017nonasymptotic, cheng2018convergence, ma2019sampling, bou2018coupling}, and other less restrictive assumptions \cite{cao2023explicit}.

The aforementioned sampling algorithms perform effectively when the target probability distribution exhibits certain niceties, such as being log-concave or unimodal. 
However, when the target distribution is multimodal, the sampling task becomes significantly more challenging. 
Even for a simple one-dimensional Gaussian mixture model, such as $0.5 N(-1, \sigma^2) + 0.5 N(1, \sigma^2)$, it has been observed that optimally tuned Hamiltonian Monte Carlo and random walk Metropolis algorithms exhibit mixing times that scale with $e^{\frac{1}{2 \sigma^2}}$ as $\sigma \to 0$ \cite{mangoubi2018does, dunson2020hastings}. 
In real-world applications, multimodal distributions are frequently encountered. 
In this context, the available theoretical guarantees for MCMC algorithms tend to be weaker, and numerical results often encounter energy barriers between modes. 
To address this challenge, various techniques have been proposed, including tempering methods \cite{swendsen1986replica, marinari1992simulated, neal2001annealed}, the equi-energy sampler \cite{holmes2006bayesian}, biasing techniques \cite{wang2001efficient, laio2002escaping}, the birth-death algorithm \cite{lu2019accelerating, tan2023accelerate}, and accelerating methods\cite{parno2018transport}, among others.
 
 In contrast to the stochastic sampling methods previously discussed, deterministic methods for sampling also exist. 
 In an ideal scenario, one can sample from any one-dimensional probability distribution by applying the inverse transform of the cumulative distribution function (CDF) of that distribution on a uniform distribution. 
 The computational cost for generating a sample in this way is typically just one function evaluation of the inverse transform. 
 However, as straightforward and efficient as this method may be, it becomes largely intractable for high-dimensional unnormalized sampling problems.
One possible approach to tackle the sampling problem is to build an optimal transport map from a reference distribution that is straightforward to sample from, such as the Gaussian distribution, to the target distribution. 
However, in practice, computing this optimal transport map can be a challenging task \cite{marzouk2016sampling,gao2020generative, alfonso2023generative}, as it often necessitates solving the Monge-Ampère equation \cite{villani2009optimal}.
 
\subsection*{Contributions}

We propose a new unit-time ODE flow called the F\"ollmer flow, which efficiently transforms a Gaussian measure into a desired target measure at time 1. 

\begin{figure}[htbp]
  \centering
  \includegraphics[width=.5\textwidth]{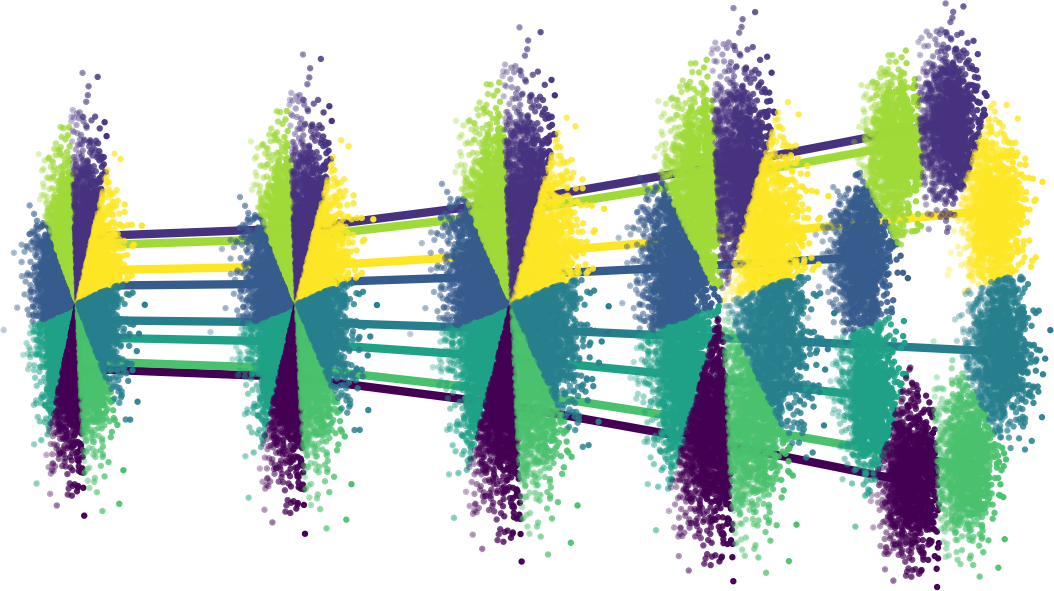}
  \caption{Time evolution of the F\"{o}llmer flow for 8-mode circle shaped Gaussian mixture, with the trajectory of particles from different modes.}
  \label{fig:evo}
\end{figure}

The Föllmer flow offers several key advantages:
\begin{itemize}
\item 
It defines a unit-time scheme without relying on ergodicity.
\item 
Numerical results demonstrate that it effectively addresses the challenge of multi-modality in sampling. See \cref{fig:evo} for illustration, and find more results in \cref{sec:experiments}.
\item 
We theoretically prove that the sampling error in the Wasserstein-2 distance is of order $\mathcal{O} \left( \sqrt[3]{\frac{d}{K^2}} \right) + \mathcal{O} \left( \frac{d}{KM} \right)$, where $K$ is the Euler discretization level, $M$ is the number of Monte Carlo samples, and $d$ is the dimension.
\item 
It sheds light on the development of one-step neural samplers, enhancing the efficiency of sampling algorithms.
\item 
It can be seamlessly integrated with existing MCMC methods as a warm starter, aiding in the recovery of collapsed modes and achieving a favorable balance between sampling efficiency and accuracy.
\end{itemize}

 \subsection*{Related work}
 Recently, the construction of transportation maps and sampling based on such maps has garnered considerable attention. 
 Dai et al. \cite{dai2023lipschitz} formally constructed the Föllmer flow and delved into the well-posedness and related properties of the flow map.
Albergo et al. \cite{albergo2023stochastic, albergo2022building} explored a unit-time normalizing flow that is relevant to the Föllmer flow. Their work takes the perspective of stochastic interpolation between the reference measure and the target measure.
Huang et al. \cite{huang2021schrodingerfollmer, jiao2021convergence} introduced the Schrödinger-Föllmer sampler (SFS), which transports the degenerate distribution at time 0 to the target distribution at time 1.
Liu and Wang \cite{liu2016stein} proposed Stein Variational Gradient Descent (SVGD), which transports particles to match the target distribution.
Liu et al. \cite{liu2022flow} presented the rectified flow, which learns neural ODE models to transport between two empirically observed distributions.
Qiu and Wang \cite{qiu2023efficient} fitted an invertible transport map between a reference measure and the target distribution.
Huang et al. \cite{huang2023monte} considered sampling from the target distribution by reversing the Ornstein-Uhlenbeck diffusion process.
Inspired by the work of Dai et al. \cite{dai2023lipschitz}, we further extended the Föllmer flow to a preconditioned version. Based on this extension, we propose a novel sampling method and establish an error bound for the proposed sampling scheme.

 \subsection*{Organization}
The remainder of this paper is structured as follows. 
In this section, we provide notations and introduce the necessary preliminaries.
In Section \ref{sec:follmer-flow}, we introduce the preconditioned version of the F\"{o}llmer flow and discuss its well-posedness and associated properties.
In Section \ref{sec:scheme}, we present various equivalent formulations of the velocity field for the ODE system and implement the F\"{o}llmer flow using Euler's method.
In Section \ref{sec:error-analysis}, we analyze the convergence of the numerical scheme under reasonable conditions (see \cref{cond:lip,cond:bound}).
Finally, in Section \ref{sec:experiments}, we conduct numerical experiments to demonstrate its performance
The derivation and establishment of the well-posedness of the F\"{o}llmer flow can be found in Appendix \ref{app:derivation}. Further insights into the error analysis of the Monte Carlo F\"ollmer flow are provided in Appendix \ref{app:analysis}. Detailed numerical settings are documented in Appendix \ref{app:settings}.

\subsection{Notations}
The space $\mathbb{R}^d$ is endowed with the standard Euclidean metric and we denote by $|\cdot|$ and $\langle \cdot, \cdot \rangle$ the corresponding norm and inner product.
Let $\mathbb S^{d-1} :=\{x\in \mathbb{R}^d: |x|=1 \}$.
The operator norm of a matrix $A \in \mathbb{R}^{m \times n}$ is denoted by $\| A \|_{\mathrm{op}} := \sup_{x\in \mathbb{R}^n, |x|=1} \|Ax\|$ and $A^{\top}$ is the transpose of $A$.
We use $\mathrm{I}_d$ to denote the $d \times d$ identity matrix.
For a twice continuously differentiable function $f$, let $\nabla f, \nabla^2 f$, and $\Delta f$ denote its gradient, Hessian, and Laplacian with respect to the spatial variable, respectively.
For $k, d, n \ge 1$, we denote by $C^n(\mathbb{R}^k; \mathbb{R}^d)$ the space of continuous functions $f: \mathbb{R}^k \to \mathbb{R}^d$ that are $n$ times differentiable and whose partial derivatives of order $n$ are continuous.
The Borel $\sigma$-algebra of $\mathbb{R}^d$ is denoted by $\mathcal{B}(\mathbb{R}^d)$.
The space of probability measures defined on $(\mathbb{R}^d, \mathcal B(\mathbb{R}^d))$ is denoted as $\mathcal P(\mathbb{R}^d)$.
For any $\mathbb{R}^d$-valued random variable $X$, we use $\mathbb{E}[X]$ and $\cov(X)$ to denote its expectation and covariance matrix, respectively.
We use $\mu * \nu$ to denote the convolution for any two probability measures $\mu$ and $\nu$.
Let $g: \mathbb{R}^k \to \mathbb{R}^d$ be a measurable mapping and $\mu$ be a probability measure on $\mathbb{R}^k$. The push-forward measure $f_{\#} \mu$ of a measurable set $A$ is defined as $f_{\#} \mu := \mu(f^{-1} (A))$.
Let $\gamma^{\mu, \Sigma}$ denote the $d$-dimensional Gaussian measure with mean vector $\mu \in \mathbb{R}^d$ and covariance matrix $\Sigma \in \mathbb{R}^{d \times d}$,
and let $\varphi^{\mu, \Sigma}(x)$ denote the probability density function of $\gamma^{\mu, \Sigma}$ with respect to the Lebesgue measure.
If $\mu = \mathbf{0}, \Sigma = \mathrm{I}_d$, we briefly write down $\gamma_d$ and $\varphi(x)$.
$\nu\left(\mathrm{d}x\right) = p\left(x\right)\mathrm{d}x$ is the target probability measure on $\mathbb{R}^{d}$.
Let $\delta_x$ denote the Dirac measure centered on some fixed point x.
For two probability measure $\mu, \nu \in \mathcal P(\mathbb{R}^d)$, we denote $\mathcal C(\mu, \nu)$ the set of transference plans of $\mu$ and $\nu$. and the Wasserstein distance of order $p \geq 1$ is defined as
\[ \mathcal W_p (\mu, \nu) := \inf_{\Pi \in \mathcal C (\mu, \nu)} \left( \int_{\mathbb{R^d} \times \mathbb{R}^d} | x-y |^p \Pi(\mathrm{d}x, \mathrm{d}y) \right)^{1/p}. \]

\subsection{Preliminaries}
\label{sec:main}
We introduce two definitions to delineate the convexity properties of probability measures and introduce relevant notations.

\begin{definition}[$\kappa$-semi-log-concave, \cite{cattiaux2014semi,mikulincer2021brownian}]
  A probability measure \[\mu \left( dx \right)  = \exp{\left( - U \left( x \right) \right)}\mathrm{d}x\] is $\kappa$-semi-log-concave
  for some $\kappa \in \mathbb{R}$ if its support $\Omega \subset \mathbb{R}^d$ is convex and $ U \in C^2 \left( \Omega \right)$ satisfies
  \[ \nabla ^2 U \left( x \right) \succeq \kappa \mathrm{I}_d, \quad\forall x \in \Omega. \]
\end{definition}

\begin{definition}[$\beta$-semi-log-convex, \cite{eldan2018regularization}]
  A probability measure \[\mu \left( dx \right)  = \exp{ \left( - U \left( x \right) \right)}\mathrm{d}x\] is $\beta$-semi-log-convex
  for some $\beta > 0$ if its support $\Omega \subset \mathbb{R}^d$ is convex and $ U \in C^2 \left( \Omega \right)$ satisfies
  \[ \nabla^2 U \left( x \right) \preceq \beta \mathrm{I}_d, \quad\forall x \in \Omega. \]
\end{definition}

We address three conditions under which we will build the F\"{o}llmer flow later.

\begin{condition}
  \label{cond:cond1}
  The probability measure $\nu$ has a finite third moment and is absolutely continuous with respect to the standard Gaussian measure $\gamma_{d}$.
\end{condition}

\begin{condition}
  \label{cond:cond2}
  The probability measure $\nu$ is $\beta$-semi-log-convex for some $\beta > 0$.
\end{condition}

\begin{condition}
  \label{cond:cond3}
  Let $D := \left(1 / \sqrt{2}\right) \diam\left(\supp\left(\nu\right)\right)$. The probability measure $\nu$ satisfies one or more of the following conditions:
  \begin{condenum}
    \item $\nu$ is $\kappa$-semi-log-concave for some $\kappa > 0$ with $D \in (0, \infty]$; \label{item:non-neg-kappa}
    \item $\nu$ is $\kappa$-semi-log-concave for some $\kappa \leq 0$ with $D \in (0, \infty)$; \label{item:neg-kappa}
    \item $\nu = N\left( \textbf{0}, \sigma^{2}\mathrm{I}_{d}\right) * \rho $ where $\rho$ is a probability measure supported on a ball of radius $R$ on $\mathbb{R}^{d}$. \label{item:conv}
  \end{condenum}
\end{condition}

\begin{remark}
  It is worth noting that \cref{item:neg-kappa} is of practical value as it corresponds to the multimodal distributions.
\end{remark}

\begin{remark}
  \Cref{item:conv} covers the Gaussian mixture examples in later numerical studies (see \cref{table:example}).
\end{remark}

\section{Preconditioned F\"{o}llmer flow}
\label{sec:follmer-flow}
{ In \cite{albergo2023stochastic,albergo2022building,gao2023gaussian}, the authors propose the stochastic interpolation}, and the F\"{o}llmer flow \cite{dai2023lipschitz} can be seen as a special case of the interpolation.
By similar techniques, we obtain a preconditioned version of the F\"{o}llmer flow which starts from a general Gaussian measure $\gamma^{\mu, \Sigma}$ instead of the standard one, where $\mu$ is the mean vector, and $\Sigma$ is a symmetric and positive semi-definite matrix that admits Cholesky decomposition $\Sigma = A A^\top$.

\subsection{Preconditioned F\"{o}llmer flow}
\begin{definition}[Preconditioned F\"{o}llmer flow]
  \label{def:fflow}
  Suppose that probability measure $\nu$ satisfies \cref{cond:cond1}. If $\left(X_{t}\right)_{t \in [0, 1]}$ solves the IVP
  \begin{equation}
    \label{eq:reverse-ivp}
    \frac{\mathrm{d}X_{t}}{\mathrm{d}t} = V\left( t, X_{t} \right), \quad X_{0} \sim \gamma^{\mu, \Sigma}, \quad t \in [0, 1].
  \end{equation}
  The velocity vector field $V$ is defined by
  \begin{equation}
    \label{eq:v-in-s}
    V(t, x) := \frac{x - \mu + \Sigma S(t, x)}{t}, \forall t \in (0, 1]; \quad V\left(x, 0\right) := \mathbb{E}_{\nu}[X] - \mu; \quad r\left(x\right) := \frac{\mathrm{d}\nu}{\mathrm{d}\gamma^{\mu, \Sigma}}\left(x\right),
  \end{equation}
  where $S(t, x)$ is the score function defined in \cref{eq:score-def}.
  We call $\left(X_{t}\right)_{t \in [0, 1]}$ a F\"{o}llmer flow and $V\left(t, x\right)$ a F\"{o}llmer velocity field associated to $\nu$, respectively.
\end{definition}

\begin{remark}[Lipschitz property]
  \label{rmk:lip-x1}
  Suppose that \cref{cond:cond1,cond:cond2,cond:cond3} hold, then the F\"{o}llmer velocity $V(t, x)$ is Lipschitz continuous with Lipschitz constant $\theta_t$,
  and the F\"{o}llmer flow \cref{eq:reverse-ivp} $X_t$ is a Lipschitz mapping at time $t=1$, with Lipschitz constant $\exp(\int_{0}^{1} \theta_s \mathrm{d}s)$.
\end{remark}

\begin{theorem}[Well-posedness]
  \label{thm:well-pose}
  Suppose that \cref{cond:cond1,cond:cond2,cond:cond3} hold. Then the F\"{o}llmer flow $\left( X_{t} \right)_{t \in [0, 1]}$ associated to $\nu$ is a unique solution to the IVP \cref{eq:reverse-ivp}.
  Moreover, the push-forward measure $\gamma^{\mu, \Sigma} \circ \left( X_{1}^{-1} \right) = \nu$.
\end{theorem}

\begin{remark}[Translation]
Various selections of $\mu$ correspond to spatial translations of the F\"{o}llmer flow and are, therefore, of minimal impact on the numerical performance.
\end{remark}

\section{Numerical scheme}
\label{sec:scheme}

We introduce several numerical schemes for F\"{o}llmer flow in this section.

\subsection{Velocity field}
We derive several advantageous equivalent representations of the velocity field, which enhance the numerical simulations.

\begin{definition}[Heat semigroup]
  Define an operator $\left(\mathcal{Q}_{t}\right)_{t \in [0, 1]}$, acting on function $r: \mathbb{R}^{d} \rightarrow \mathbb{R}$ by
  \begin{equation*}
    \begin{aligned}
      \mathcal{Q}_{1-t} r \left( x \right) & := \int_{\mathbb{R}^{d}} \varphi^{tx + (1-t)\mu, (1-t^2)\Sigma}\left( y \right) r\left( y \right) \mathrm{d}y \\
                                           & = \int_{\mathbb{R}^{d}} r\left( tx + (1-t)\mu + \sqrt{1-t^{2}} Az \right) \mathrm{d}\gamma_d \left(z\right).
    \end{aligned}
  \end{equation*}
\end{definition}

\begin{remark}
  Notice that
  \[
    \mathcal{Q}_{1-t} r(x) = \frac{1}{\varphi^{\mu, \Sigma}(x)} \times \int_{\mathbb{R}^d} \varphi^{ty+(1-t)\mu, (1-t^2)\Sigma}(x)p(y) \mathrm{d}y.
  \]
  Hence, we obtain
  \[
    x - \mu + \Sigma S(t, x) = \Sigma \nabla \ln \mathcal{Q}_{1-t} r(x), \quad \forall t \in (0, 1].
  \]
  Therefore, the velocity field on time interval $ t \in (0, 1]$ can be interpreted as
  \begin{equation}
    \label{eq:v-in-q}
    V(t, x) = \frac{\Sigma \nabla \ln \mathcal{Q}_{1-t} r\left(x\right)}{t}.
  \end{equation}
\end{remark}

Notice that the operator $\mathcal{Q}$ admits property
\[
  \nabla \mathcal{Q}_{1-t} r(x) = t \mathcal{Q}_{1-t} \nabla r(x),
\]
then by direct calculation, the velocity field \cref{eq:v-in-q} yields
\begin{equation}
  \label{eq:v-mc-raw}
  V(t, x) = \frac{\Sigma \mathcal{Q}_{1-t} \nabla r(x)}{\mathcal{Q}_{1-t} r(x)}, \quad t \in (0, 1].
\end{equation}

By Stein's lemma, we can avoid the calculation of $\nabla r$ in \cref{eq:v-mc-raw}, that is
\begin{equation}
  \label{eq:stein}
  V(t, x) = \frac{A \int_{\mathbb{R}^{d}} z r(z) \mathrm{d} \gamma^{tx+(1-t)\mu, (1-t^2)\Sigma}(z)}{\sqrt{1-t^2} \int_{\mathbb{R}^{d}} r(z) \mathrm{d} \gamma^{tx+(1-t)\mu, (1-t^2)\Sigma}(z)}, \quad t \in (0, 1).
\end{equation}

Since the velocity $V(t, x)$ is scale-invariant with respect to $r$, the F\"{o}llmer flow can be used for sampling from target measure taking form
\begin{equation*}
  \nu \left( \mathrm{d}x \right) = \frac{1}{C} \exp \left( - U \left( x \right) \right) \mathrm{d}x,
\end{equation*}
where $C$ may be unknown.

In general, the F\"{o}llmer velocity $V(t, x)$ stated in \cref{eq:stein} does not have a closed-form expression.
But fortunately, it is compatible with Monte Carlo approximations.
Let $Z_{1}, \dots, Z_{M}$ be i.i.d. $N\left( \textbf{0}, \mathrm{I}_{d} \right)$, where $m$ is sufficiently large. We can approximate $V(t, x)$ by
\begin{equation}
  \label{eq:mc-velocity}
  \tilde{V}\left( t, x \right) := \frac{A \sum_{j=1}^{M} \left[ Z_{j} r \left( tx + (1-t)\mu +
  \sqrt{1-t^2}AZ_{j} \right) \right]}{\sqrt{1-t^2} \cdot \sum_{j=1}^{M}\left[ r \left( tx + (1-t)\mu +  \sqrt{1-t^2}AZ_{j} \right) \right] }
\end{equation}

Vargas et al. \cite{vargas2022bayesian} observe that the term $r$ involves the product of density functions evaluated at samples, rather than a log product, and is thus prone to numerical instability. 
They introduce a stable implementation of \cref{eq:mc-velocity} by leveraging the logsumexp technique and properties of the Lebesgue integral. 
We employ the logsumexp reformulation, using similar techniques, in our numerical studies.

\subsection{Euler's method}
\Cref{thm:well-pose} indicates that we can initiate the process with $x_{0} \sim \gamma^{\mu, \Sigma}$ and update the values of $\left\{ x_{t}: 0 < t \leq 1\right\}$ according to the continuous-time "{o}llmer flow, as described by \cref{eq:reverse-ivp}. In this manner, the distribution of $x_{1}$ precisely converges to $\nu$.
For the discretization of the ODE system in \cref{eq:reverse-ivp}, we employ Euler's method with a fixed step size. To circumvent potential numerical instability issues at $t=0$ and $t=1$, we introduce a truncation of the unit-time interval by $\varepsilon$ at both endpoints.

Let
\[
  t_{k} = \varepsilon + k h,\  k=0, 1, \dots, K, \ \text{with} \  h = (1-2\varepsilon)/K,
\]
and draw $x_{0}$ from $\gamma^{\mu, \Sigma}$. Then \cref{eq:reverse-ivp} using Euler's method has the form
\begin{equation}
  \label{eq:ode-euler}
  x_{t_{k+1}} = x_{t_{k}} + h V\left(t_{k}, x_{t_{k}}\right), \ k=0, 1, \dots, K-1.
\end{equation}

The pseudocode for implementing \cref{eq:ode-euler} is presented in \cref{alg:close}.
\begin{algorithm}
  \caption{F\"{o}llmer flow for $\nu = \exp\left( -U \left( x \right) \right) / C$}
  \label{alg:close}
  \begin{algorithmic}
    \REQUIRE{$U\left( x \right), K$. Initialize $h = 1/K, x_{t_0} \sim \gamma^{\mu, \Sigma}$.}
    \FOR{$k=0, 1, \dots, K-1$}
    \STATE{Compute the velocity $V\left( t_{k}, x_{t_{k}} \right)$ by \cref{eq:v-in-s} or \cref{eq:v-in-q}, }
    \STATE{Update $x_{t_{k+1}} = x_{t_{k}} + hV\left( t_{k}, x_{t_{k}}\right)$. }
    \ENDFOR
    \ENSURE{$\{ x_{t_{k}} \}_{k=1}^{K} $.}
  \end{algorithmic}
\end{algorithm}

\section{Error analysis}
\label{sec:error-analysis}
In this section we derive the error analysis of the F\"{o}llmer flow. We use $\tilde{X}_{t_k}$ to bridge between the continuous version $X_t$ and the Monte Carlo version $\tilde{Y}_{t_k}$.
Without loss of generality, we assume that $\mu = \mathbf{0}$ and $\Sigma = \mathrm{I}_d$ in the following analysis.
In our setting, $X_0, \tilde{X}_{\varepsilon}$ and $\tilde{Y}_{\varepsilon}$ are drawn from $\gamma^{\mu,\Sigma}$.
Our interest lies in the distance between $X_1$ and $\tilde{Y}_{1-\varepsilon}$.

We first address the following assumptions on the density ratio $r$.

\begin{condition}
  \label{cond:lip}
  $r$ and $\nabla r$ are Lipschitz continuous with constant $\gamma$.
\end{condition}

\begin{condition}
  \label{cond:bound}
  There exists $\xi>0$ such that $r \geq \xi$.
\end{condition}

We present our main theorem here. Details of proof are given at \cref{app:analysis}.
\begin{theorem}
  \label{thm:main}
  Suppose \cref{cond:cond1,cond:cond2,cond:cond3,cond:lip,cond:bound} hold,
  by choosing the truncation step to be $\mathcal{O} \left( \sqrt[3]{\frac{d}{K^2}} \right)$, the error of the Monte Carlo F\"{o}llmer flow using Euler's method is given by
  \[
    \mathcal{W}_2 (\law(\tilde{Y}_{1-\varepsilon}), \law(X_1)) \leq \mathcal{O} \left( \sqrt[3]{\frac{d}{K^2}} \right) + \mathcal{O} \left( \frac{d}{KM} \right).
  \]
\end{theorem}

\begin{corollary}
  \label{cor:cor1}
  Suppose \cref{cond:cond1,cond:cond2,cond:cond3,cond:lip,cond:bound} hold,
  by choosing the truncation step to be $ \mathcal{O} \left( \sqrt[3]{\frac{d}{K^2}} \right) $ and $ M = \mathcal{O} \left( \sqrt[3]{\frac{d^2}{K}} \right) $,
  the error of the Monte Carlo F\"{o}llmer flow using Euler's method is given by
  \[
    \mathcal{W}_2 (\law(\tilde{Y}_{1-\varepsilon}), \law(X_1)) \leq \mathcal{O} \left( \sqrt[3]{\frac{d}{K^2}} \right).
  \]
\end{corollary}

\Cref{thm:main} indicates that for appropriate truncation step size $\varepsilon$,
the overall error tends to zero as the number of Monte Carlo simulations $M$ and the number of temporal grid $K$ tend to infinity.
\Cref{cor:cor1} further indicates that for appropriate $M$, the overall error tends to zero as $K$ tends to infinity.

\section{Numerical experiments}
\label{sec:experiments}
In this section, we undertake numerical experiments to assess the performance of the F\"{o}llmer flow. 
We employ Markov Chain Monte Carlo (MCMC) methods, including the Metropolis-Hastings algorithm (MH), the tamed Metropolis-adjusted Langevin algorithm (MALA), and the tamed unadjusted Langevin algorithm (ULA), using 50 chains for the purpose of comparison. 
For reference, our code is accessible online\footnote{\url{https://github.com/burning489/SamplingFollmerFlow}}.

\subsection*{Gaussian mixture distributions}
We first derive the F\"{o}llmer flow for Gaussian mixture distributions.

Assume that the target distribution $\nu$ is a Gaussian mixture
\begin{equation}
  \label{eq:gm-def}
  \nu = \sum_{i=1}^\kappa \theta_i \gamma^{\mu_i, \Sigma_i}, \sum_{i=1}^\kappa \theta_i = 1 \text{ and } \theta_i \in [0, 1], i=1, \cdots, \kappa,
\end{equation}
where $\kappa$ is the number of mixture components, $\gamma^{\mu_i, \Sigma_i}$ is the $i$-th Gaussian component
with mean $\mu_i \in \mathbb{R}^{d}$ and covariance matrix $\Sigma_i \in \mathbb{R}^{d \times d}$.
Obviously, the target distribution $\nu$ is absolutely continuous with respect to the $d$-dimensional Gaussian distribution $\gamma^{\mu, \Sigma}$.
The density ratio is
\begin{equation}\label{eq:rnd-gm}
  r =  \frac{1}{\gamma^{\mu, \Sigma}} \sum_{i=1}^\kappa\theta_i \gamma^{\mu_i, \Sigma_i}.
\end{equation}

\subsection*{Closed-form}
By Substituting \cref{eq:gm-def} into \cref{eq:v-in-s}, we obtain the F\"{o}llmer velocity for Gaussian mixtures
\begin{equation*}
  \begin{aligned}
    V(t, x) & = \frac{1}{t} \left( x - \mu + \Sigma \frac{\sum_{i=1}^{\kappa}
    \theta_i \nabla p_i(x)}{\sum_{i=1}^{\kappa} \theta_i p_i(x)}\right)                \\
            & = \frac{1}{t} \left( x - \mu + \Sigma \frac{\sum_{i=1}^{\kappa} \theta_i
        p_i(x) \left[ t^2 \Sigma_i + (1-t^2)\Sigma \right]^{-1} \left[ t\mu_i +
          (1-t)\mu - x\right]}{\sum_{i=1}^{\kappa} \theta_i p_i(x)}\right) ,
  \end{aligned}
\end{equation*}
where $p_i$ is density function of $N(t\mu_i + (1-t)\mu, t^2\Sigma_i + (1-t^2)\Sigma)$.

\subsection*{Stable Monte Carlo form}
We can obtain a stable Monte Carlo expression of F\"{o}llmer velocity by directly plugging the Radon-Nikodym derivative \cref{eq:rnd-gm} into \cref{eq:mc-velocity}.

\subsection{One-dimensional Gaussian mixture distribution}
\label{subsec:1d}
We employ the proposed F\"{o}llmer flow alongside other methods to generate 10,000 samples from three one-dimensional Gaussian mixture distributions (as detailed in \cref{table:example}, examples 1 to 3). 
We assess the quality of the samples by computing the sampling errors, utilizing the Wasserstein-2 distance and the maximum mean discrepancy (MMD) \cite{gretton2006kernel} between the generated samples and the ground truth.
In accordance with the methodology outlined in \cite{qiu2023efficient}, we adopt the adjusted Wasserstein distance and adjusted MMD, both of which can assume negative values. Smaller values of these metrics indicate higher sample quality. 
In \cref{fig:gm1d}, we present kernel density estimation curves for all methods in red, with the target density functions shaded in grey. 
We summarize the sampling errors in \cref{table:gm1d-mcmc-error}, revealing that the F\"{o}llmer flow consistently outperforms other methods.

\begin{figure}[htpb] \centering
  \includegraphics[width=\textwidth]{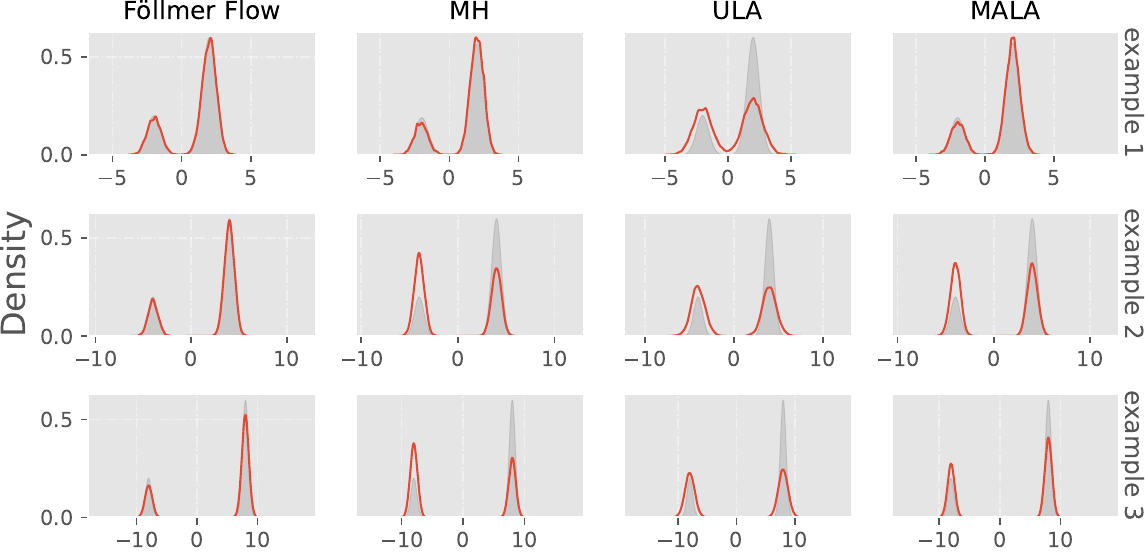}
  \caption{Kernel density estimation of 1D Gaussian mixture distributions example 1 to 3.}
  \label{fig:gm1d}
\end{figure}

\begin{table}[htb]
\caption{Adjusted 2-Wasserstein distance and adjusted MMD for samples generated by different methods for example 1 to 3.}
\label{table:gm1d-mcmc-error}
\resizebox*{\textwidth}{!}{
\begin{tabular}{lrrrrrrrr}
\toprule
example & \multicolumn{2}{|c|}{Föllmer Flow} & \multicolumn{2}{|c|}{MH\_50} & \multicolumn{2}{|c|}{tULA\_50} & \multicolumn{2}{|c|}{tMALA\_50} \\
 & adj. W & adj. MMD & adj. W & adj. MMD & adj. W & adj. MMD & adj. W & adj. MMD \\
\midrule
example 1 & \bfseries -0.001 & \bfseries -0.000 & 0.121 & 0.025 & 0.790 & 0.613 & 0.068 & 0.013 \\
example 2 & \bfseries 0.056 & \bfseries 0.005 & 2.296 & 5.372 & 2.028 & 3.924 & 1.957 & 3.916 \\
example 3 & \bfseries 0.129 & \bfseries 0.041 & 4.888 & 26.057 & 3.668 & 14.622 & 2.312 & 6.372 \\
\bottomrule
\end{tabular}
}
\end{table}

We depict the kernel density estimation curves for all methods across the various target distributions in \cref{fig:gm1d}. The target density functions are highlighted in grey for reference.
In situations where the centroids of the Gaussian components are close, the proposed F\"{o}llmer flow, Metropolis-Hastings (MH), and the tamed Metropolis-adjusted Langevin algorithm (MALA) demonstrate comparable performance. However, the tamed unadjusted Langevin algorithm (ULA) fails to capture the variance information adequately.
Conversely, when the centroids of the Gaussians move apart from each other, only samples generated using the F\"{o}llmer flow accurately represent the underlying target distribution, while all other methods exhibit a decrease in accuracy.

\subsection{Two-dimensional Gaussian mixture distribution}
\label{subsec:2d-closed}
We consider seven instances of two-dimensional Gaussian mixture distributions, as described in appendix \cref{table:example}, examples 4 to 10. 
In examples 4 and 5, the centroids of the Gaussian components form a circle, while for examples 6 to 9, they are arranged in a square matrix. 
Example 10 comprises four anisotropic Gaussian components forming a square configuration.
To evaluate these distributions, we employ the proposed F\"{o}llmer flow and other sampling methods to generate 20,000 samples, visualizing the results with kernel density estimation in \cref{fig:gm2d-cmp}. 
As depicted in \cref{fig:gm2d-cmp}, only the samples generated using the closed-form F\"{o}llmer flow accurately estimate the underlying target distribution. 
In contrast, Metropolis-Hastings (MH) and the tamed Metropolis-adjusted Langevin algorithm (MALA) tend to collapse onto one or a few modes when the distributions become more challenging to sample from. The tamed unadjusted Langevin algorithm (ULA) can capture all modes but exhibits worse variance.
The sampling error is reported in \cref{table:gm1d-mcmc-error}, demonstrating that the family of F\"{o}llmer flow samplers consistently outperforms other methods across all examples.

\begin{figure}[htpb] \centering
  \includegraphics[width=\textwidth]{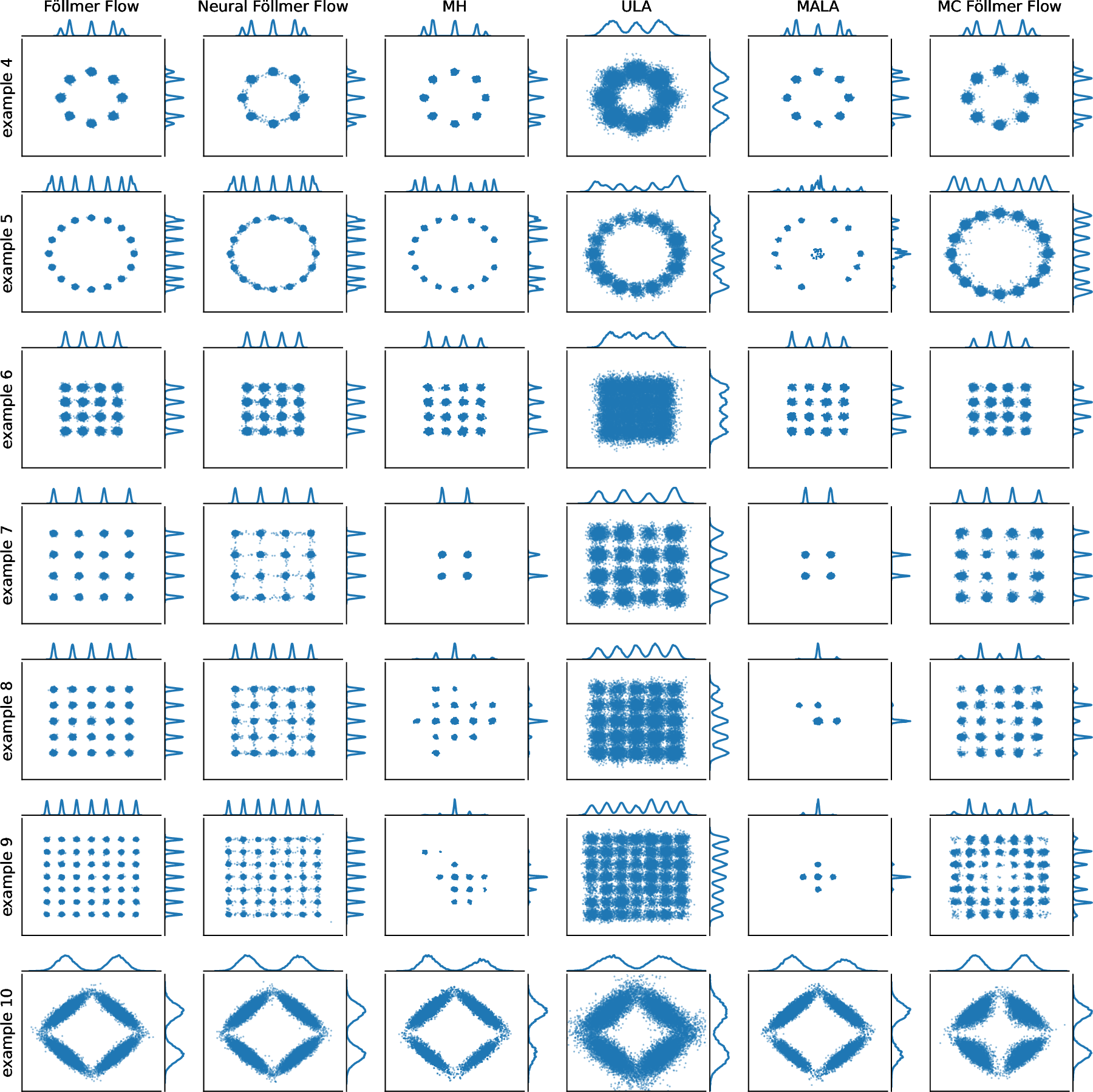}
  \caption{Kernel density estimation with marginal distribution plots for example 4 to 10.}
  \label{fig:gm2d-cmp}
\end{figure}

\begin{table}[htb]
\caption{Adjusted 2-Wasserstein distance and adjusted MMD for samples generated by different methods for example 4 to 9.}
\label{table:gm2d-mcmc-error}
\resizebox*{\textwidth}{!}{
\begin{tabular}{lrrrrrrrrrrrr}
\toprule
example & \multicolumn{2}{|c|}{F\"ollmer Flow} & \multicolumn{2}{|c|}{Neural F\"ollmer Flow} & \multicolumn{2}{|c|}{MC F\"ollmer Flow} & \multicolumn{2}{|c|}{MH\_50} & \multicolumn{2}{|c|}{tULA\_50} & \multicolumn{2}{|c|}{tMALA\_50} \\
 & adj. W & adj. MMD & adj. W & adj. MMD & adj. W & adj. MMD & adj. W & adj. MMD & adj. W & adj. MMD & adj. W & adj. MMD \\
\midrule
example 4 & \bfseries -0.024 & \bfseries -0.002 & -0.020 & -0.001 & 0.182 & -0.000 & 0.670 & 0.515 & 0.583 & 0.008 & 0.436 & 0.128 \\
example 5 & \bfseries -0.081 & 0.000 & -0.035 & 0.004 & 0.893 & \bfseries -0.001 & 0.665 & 0.115 & 1.670 & 2.127 & 4.018 & 0.673 \\
example 6 & \bfseries -0.039 & -0.002 & -0.032 & -0.001 & 0.260 & \bfseries -0.004 & 0.460 & 0.257 & 0.429 & 0.005 & 0.502 & 0.177 \\
example 7 & -0.027 & -0.013 & \bfseries -0.035 & \bfseries -0.018 & 0.710 & -0.007 & 3.178 & 0.113 & 0.807 & 0.243 & 3.106 & -0.017 \\
example 8 & \bfseries -0.023 & -0.006 & -0.014 & \bfseries -0.007 & 1.089 & 0.003 & 2.942 & -0.007 & 1.030 & 0.845 & 4.491 & 0.115 \\
example 9 & \bfseries -0.063 & 0.005 & -0.044 & 0.005 & 0.994 & 
\bfseries -0.000 & 5.515 & 0.617 & 1.265 & 0.845 & 6.389 & 0.117 \\
example 10 & -0.033 & -0.002 & \bfseries -0.038 & \bfseries -0.004 & 0.178 & 0.008 & 1.192 & 1.781 & 0.429 & 0.011 & 0.876 & 0.910 \\
\bottomrule
\end{tabular}
}
\end{table}

\subsection{Preconditioned F\"{o}llmer flow}
\label{subsec:precondition}
We investigate the process of sampling from example 7 using the Monte Carlo F\"ollmer flow. 
It is evident that the selection of a suitable preconditioner, denoted as $\Sigma$, significantly impacts sample quality. 
We initialize the covariance matrix of the initial Gaussian measure as $\Sigma = \sigma^2 \mathrm{I}_2$, where $\sigma$ varies between 1.2 and 3.0 with increments of 0.4. 
For the sake of simplicity, we fix the mean of the initial Gaussian measure at $\mu = \mathbf{0}$.

We generate 20,000 samples for each preconditioner. 
The effects of different choices for $\Sigma$ are presented in \cref{fig:gm2d-precondition}. 
Lower variance settings result in a loss of outer modes, while higher variance settings lead to the loss of inner modes. 
Consequently, an appropriate preconditioner enhances the performance of the sampler. 
The guiding principle for selecting the variance is to align it closely with that of the target distribution.

\begin{figure}[htpb] \centering
  \includegraphics[width=\textwidth]{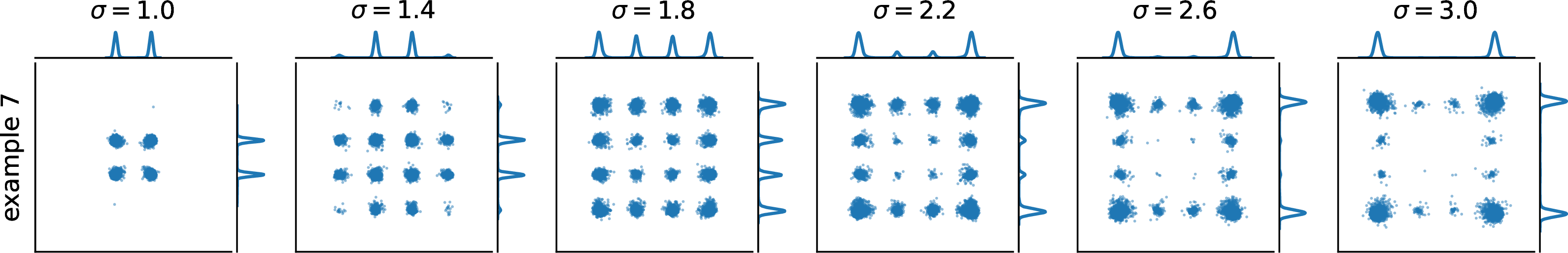}
  \caption{Kernel density estimation of preconditioned F\"{o}llmer flow samples with marginal distribution plots for example 7.}
  \label{fig:gm2d-precondition}
\end{figure}

\subsection{MC F\"{o}llmer flow}
\label{subsec:mc}
In general, computing the closed-form F\"{o}llmer velocity can be a challenging task, whereas the Monte Carlo version presented in \cref{eq:mc-velocity} is a more practical and manageable alternative for implementation.

\subsubsection{On 2-dimensional mixtures}
\label{subsubsec:2d-mc}
We conducted experiments with the Monte Carlo F\"{o}llmer flow on example 4 to 10, as detailed in \cref{table:example}, generating 20,000 samples. 
The visualizations of kernel density estimation in the rightmost column of \cref{fig:gm2d-cmp} illustrate that the Monte Carlo F\"{o}llmer flow produces samples of relatively high quality. 
This outcome suggests that the Monte Carlo formulation is a practical approach and can be extended to handle more complex cases.

\subsubsection{On hybridizing with MCMC methods}
\label{subsubsec:hybrid}
The previous results have demonstrated that the proposed F\"{o}llmer flow is an effective sampler suitable for addressing multimodal problems. 
However, it is computationally more resource-intensive compared to classical MCMC methods. 
On the other hand, classical MCMC methods are faster but may encounter mode losses.
To tackle this challenge, we explore the idea of hybridizing the F\"{o}llmer flow with existing MCMC methods. 
We use a small number of samples generated by the F\"{o}llmer flow as the initial particles for MCMC samplers. 
In this setup, the F\"{o}llmer flow serves as a warm start sampler, while the MCMC methods play the primary role in generating samples. 
This predictor-corrector scheme enables us to leverage the high-quality sampling capabilities of the F\"{o}llmer flow and the speed of classical MCMC methods simultaneously.

In order to demonstrate this approach, we conducted an experiment on example 7, as detailed in appendix \cref{table:example}. 
The results are presented in \cref{fig:gm2d-hybrid}. 
For the sake of efficiency, we set $M=K=10$ for the Monte Carlo F\"{o}llmer flow.
In \cref{fig:gm2d-hybrid}, we compare the outcomes of generating 20,000 samples using the original MCMC method and the hybrid method. 
The results indicate that the F\"{o}llmer flow assists Metropolis-Hastings (MH) and the tamed Metropolis-adjusted Langevin algorithm (MALA) in capturing more modes. 
The tamed unadjusted Langevin algorithm (ULA) also benefits from the hybrid approach, as it results in a reduction in the variance of the generated samples compared to those produced by the original method. 
Remarkably, such improvements can be achieved with only a preliminary trial of the Monte Carlo F\"ollmer flow, utilizing a small number of MC samples and a sparse time grid.

\begin{figure}[htpb] \centering
  \includegraphics[width=.6\textwidth]{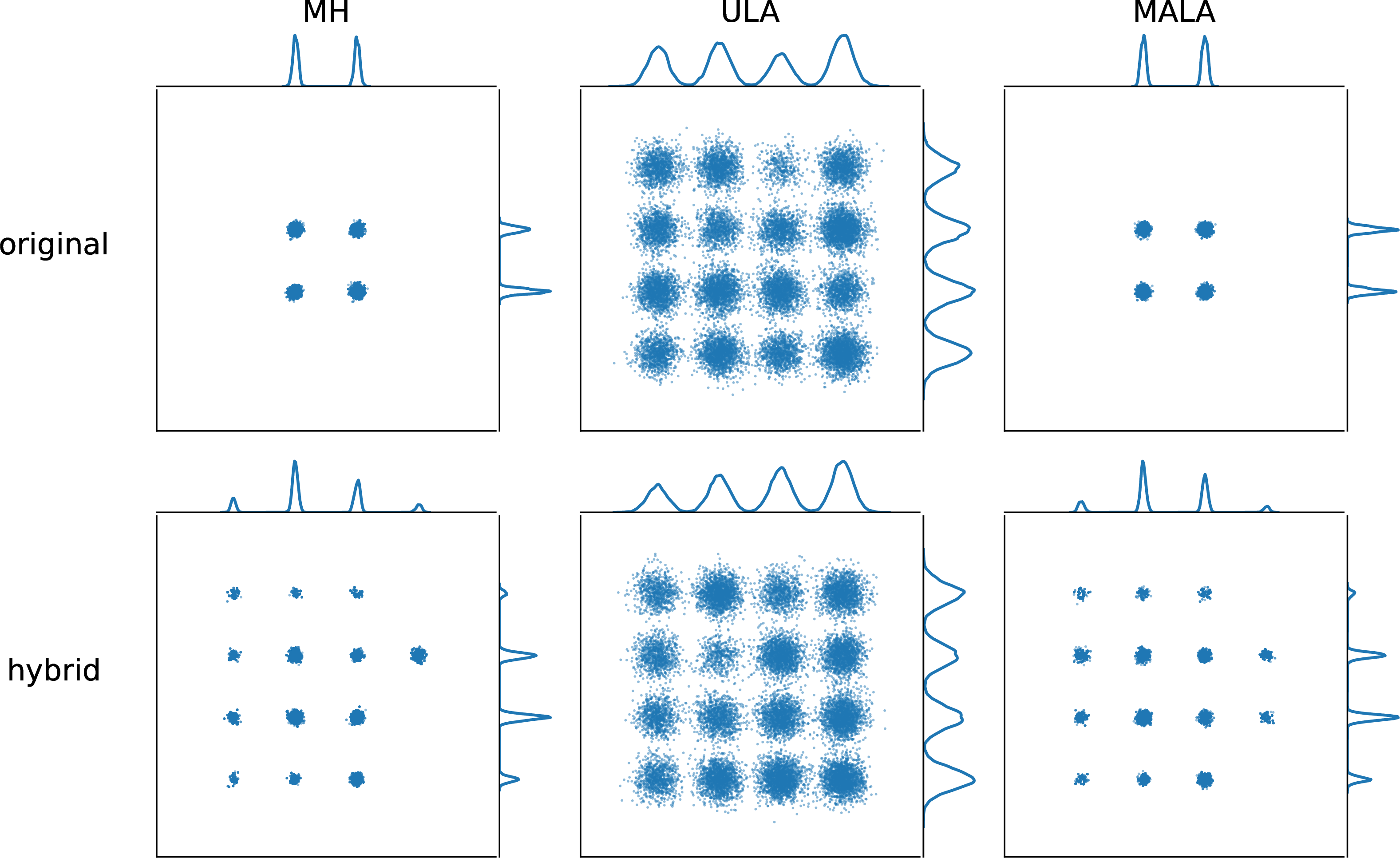}
  \caption{Kernel density estimation of Monte Carlo F\"{o}llmer flow samples with marginal distribution plots for the 16-mode square shaped Gaussian mixture distributions.
  Each column shows the samples generated by original MCMC method and hybrid version warm-stared by F\"{o}llmer flow.}
  \label{fig:gm2d-hybrid}
\end{figure}

\subsubsection{On higher dimensional mixtures}
\label{subsubsec:nd-mc}
We conducted experiments on a two-mode Gaussian mixture distribution, as described in \cref{table:example} example 11, with dimensions ranging from 1 to 10. 
It's worth noting that the score function $S(t, x)$ is related to the convolution distribution $[t \nu] * [(1-t^2) \varphi_d]$. 
As a consequence, the variance of the convoluted $\varphi_d$ decreases as time approaches 1. 
This makes estimating the velocity field via Monte Carlo approximation more challenging as time approaches 0.
To address this issue, we introduce a non-uniform time discretization by setting
\[ T_n = 1 - \exp\left(-T_u\right), \]
where $T_u$ is the uniform grid on the interval $[\varepsilon, +\infty)$. 
This adjustment helps alleviate the numerical instability around $t = \varepsilon$ and allows for more iterations around $t = 1-\varepsilon$.
We consider four test functions, denoted as $h (x)$, which include the first moment $h(x) = \alpha^\top x$, the second moment $h(x) = (\alpha^\top x)^2$, the moment generating function $h(x) = \exp(\alpha^\top x)$, and $h(x) = 5 \cos(\alpha^\top x)$, with $\alpha \in \mathbb{R}^d$ satisfying $\left\| \alpha \right\|_2 = 1$. 
These test functions serve as additional evaluation criteria for our experiments.

In the case of the Monte Carlo F\"{o}llmer flow, we set the number of Monte Carlo samples, denoted as $M$, to grow linearly with the dimension $d$, with the relationship $M=200d$. 
We generate 20,000 samples using each method.
Upon comparing the Monte Carlo estimates of $\mathbb{E}[h(X)]$ across different samplers and comparing them to the ground truth, we observe that the closed-form F\"{o}llmer flow consistently performs well and exhibits greater stability than other methods. 
While the Monte Carlo F\"{o}llmer flow is somewhat less stable and accurate compared to the closed-form version, it achieves performance similar to that of the MCMC methods. 
These results indicate that for the Monte Carlo F\"{o}llmer flow, the number of Monte Carlo samples required grows linearly with the dimension, suggesting that it does not suffer from the curse of dimensionality in this case.

\begin{figure}[htpb] \centering
  \includegraphics[width=\textwidth]{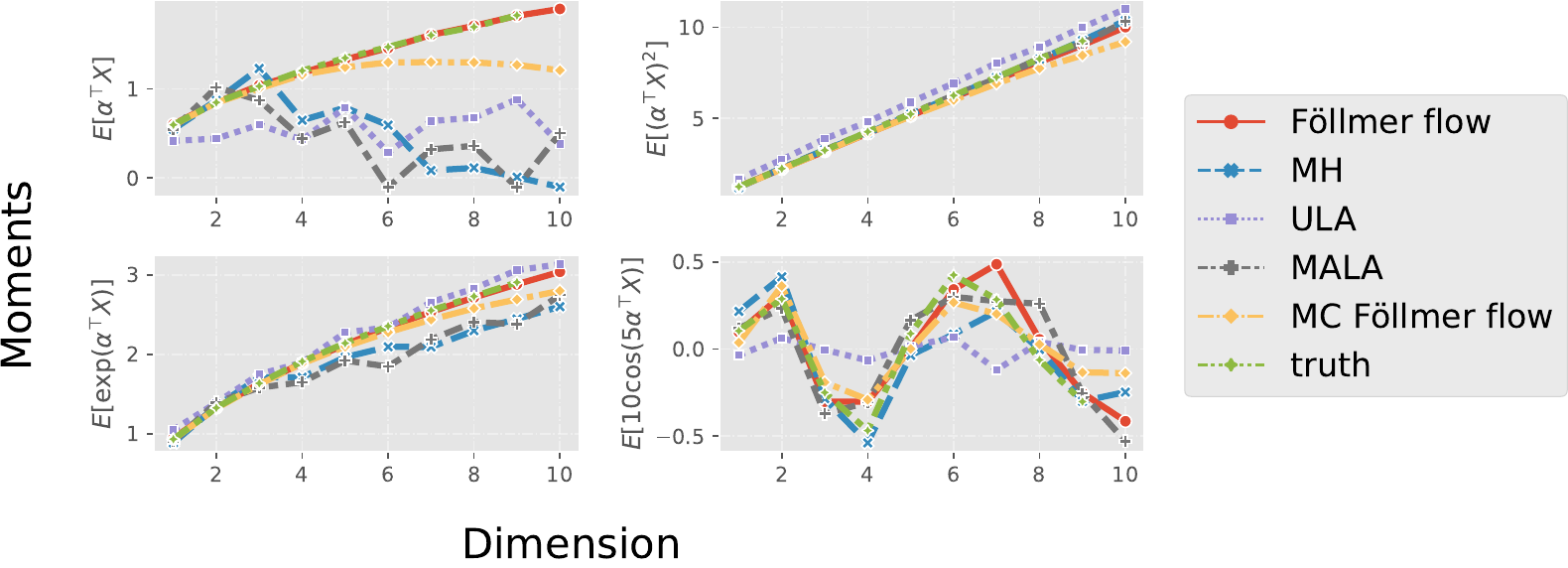}
  \caption{Monte Carlo estimates of $\mathbb{E} \left[ h \left( x \right) \right]$ versus $d$ for $d$-dimensional multivariate Gaussian mixture distributions of $X$,
    for d increasing from 1 to 10 with lag 1.} \label{fig:gmnd-moment}
\end{figure}

\subsection{Neural F\"{o}llmer flow}
\label{subsec:nff}
The SDE-based sampling method SFS, as described in \cite{huang2021schrodingerfollmer,jiao2021convergence}, attains sample quality similar to that of the proposed ODE flow. 
However, it's important to note that the F\"{o}llmer flow is deterministic for each particle and possesses a non-degenerate initial distribution.
This property enables the possibility of approximating the Lipschitz map of the ODE sampler $X_1$ using a deep neural network. 
Such an approximation could potentially enhance the efficiency and versatility of the F\"{o}llmer flow in practice.

Indeed, by approximating the Lipschitz map of the ODE sampler $X_1$ with a deep neural network, we can achieve significant computational cost savings. 
Recall that using Euler's method, the Monte Carlo F\"{o}llmer flow has a computational complexity of $\mathcal{O} \left( K M \# r \right)$, where $K$ represents the grid size, $M$ is the number of Monte Carlo samples, and $\# r$ is the cost of evaluating the RND $r$, which is linear in $d$ at best. 
As we refine the temporal grid, the computational cost becomes increasingly expensive.
In contrast, if we employ a deep neural network to approximate the map at $X_1$, the computational cost for generating one sample would be reduced to just a one-step evaluation of the network. 
This approach offers a significant enhancement in performance compared to the original scheme. 
Furthermore, the non-degenerate initial distribution ensures the richness of the generated samples, enhancing their quality.

We employ a vanilla ResNet to simulate the F\"{o}llmer flow for example 4 to 10. 
For each target distribution, we first draw $N=100,000$ samples $\{ \tilde{X}_{\varepsilon}^{(i)}\}_{i=1}^N$ from $\gamma^{\mu, \Sigma}$ and then generate corresponding F\"{o}llmer samples $\{ \tilde{X}_{1-\varepsilon}^{(i)}\}_{i=1}^N$.
Let $\mathcal F$ denote the ResNet class, and then the network is trained by minimizing the empirical risk:
\[
  \hat{f_{\theta}} \in \argmin_{f \in \mathcal F} \hat{\mathcal{L}}\left( f \right) := \sum_{i=1}^N \left\| f(\tilde{X}_{\varepsilon}^{(i)}) - \tilde{X}_{1-\varepsilon}^{(i)} \right\|.
\]
The scatter plot and marginal KDE plot of 20,000 neural F\"{o}llmer flow samples are presented in the second column of \cref{fig:gm2d-cmp}. 
These results reveal that the neural F\"{o}llmer flow achieves nearly the same sample quality as the training set, demonstrating the effectiveness of the approach.

\section{Conclusions}
\label{app:conclusions}
In summary, we have introduced a well-posed ODE flow known as the F\"{o}llmer flow, which effectively transports a Gaussian measure to the target distribution over a unit-time interval. 
This flow serves as a high-quality sampler, and our numerical results demonstrate its practicality in Monte Carlo simulations. 
We have also emphasized the importance of choosing an appropriate preconditioner for achieving optimal numerical results. 
Despite its quality, it is worth noting that the computational cost of the F\"{o}llmer flow is relatively high compared to classical MCMC methods. 
To address this issue, we have explored hybridizing our method with MCMC techniques, achieving a balance between sample quality and computational efficiency.
Additionally, we have shown that it is feasible to employ deep neural networks to simulate the F\"{o}llmer flow. 
This approach significantly reduces the computational cost to a single-step evaluation of the network while maintaining high sample quality.

Looking ahead, there are several avenues for future research. 
Investigating the impact of the temporal grid on the F\"{o}llmer flow's computational performance is an area of interest. 
We have implemented the F\"{o}llmer flow using Euler's method in this work, but exploring higher-order Runge-Kutta methods may lead to faster convergence. 
Furthermore, we have used a static step size in our work, but we empirically observe instability as time approaches 0. 
Adaptive step-size strategies may help mitigate this issue. 
Finally, conducting a convergence analysis for the neural F\"ollmer flow, especially for specific network classes, is a promising direction for further research.

\appendix

\section{Derivation and well-posedness}
\label{app:derivation}
\subsection{Diffusion process \texorpdfstring{\cite{follmer1986time,dai2023lipschitz}}{}}
For any $\varepsilon \in (0, 1)$, we consider a diffusion process $(\overline{X}_t)_{t \in [0, 1-\varepsilon]}$ defined by the following It\^{o} SDE
\begin{equation}
  \label{eq:sde}
  \mathrm{d} \overline{X}_t = -\frac{1}{1-t} (\overline{X}_t - \mu) \mathrm{d}t + \sqrt{\frac{2}{1-t}}A \mathrm{d}\overline{W}_t, \quad \overline{X}_0 \sim \nu, \quad t \in [0, 1-\varepsilon].
\end{equation}
The diffusion process defined in \cref{eq:sde} has a unique strong solution on $[0, 1-\varepsilon]$.
The transition probability distribution of \cref{eq:sde} from $\overline{X}_0$ to $\overline{X}_t$ is given by
\begin{equation*}
  \label{eq:transition}
  \overline{X}_t | \overline{X}_0 = x_0 \sim N( (1-t)x_0 + t \mu, t(2-t)\Sigma),
\end{equation*}
for every $t \in [0, 1-\varepsilon]$.

Note that the marginal distribution flow $(\overline{\mu}_t)_{t \in [0, 1-\varepsilon]}$ of the diffusion process \cref{eq:sde} satisfies the Fokker-Planck-Kolmogorov equation in an Eulerian framework
\begin{equation*}
  \partial_t \overline{\mu}_t = \nabla \cdot (\overline{\mu}_t V(1-t, x)) \quad \text{on } [0, 1-\varepsilon] \times \mathbb{R}^d, \overline{\mu}_0 = \nu,
\end{equation*}
in the sense that $\overline{\mu}_t$ is continuous in $t$ under the weak topology, where the velocity field is defined by
\begin{equation*}
  V(1-t, x) := \frac{1}{1-t}[x - \mu+ \Sigma S(1-t, x)], \quad t \in [0, 1-\varepsilon],
\end{equation*}
and
\begin{equation*}
  S(t, x) := \nabla \ln \int_{\mathbb{R}^d} \varphi^{ty + (1-t)\mu, (1-t^2)\Sigma}(x) p(y) \mathrm{d}y, \quad \forall t \in [\varepsilon, 1].
\end{equation*}
Due to the Cauchy-Lipschitz theory with smooth velocity, we shall define a flow $(X^*_t)_{t \in [0, 1-\varepsilon]}$ in a Lagrangian formulation via the following ODE system
\begin{equation*}
  \mathrm{d} X^*_t = -V(1-t, X^*_t)\mathrm{d}t, \quad X^*_0 \sim \nu, \quad t \in [0, 1-\varepsilon].
\end{equation*}

\begin{proposition}
  \label{prop-convergence}
  Assume that $V \in L^1([\varepsilon, 1]; W^{1, \infty}(\mathbb{R}^d; \mathbb{R}^d))$.
  Then the push-forward measure associated with the flow map $X_t^*$ satisfies $X_t^* := (1-t)X + t\mu + \sqrt{t(2-t)}AY$ with $X \sim \nu, Y \sim \gamma_d$.
  Moreover, the push-forward measure $\nu \circ (X_{1-\varepsilon}^*)^{-1}$ converges to the Gaussian measure $\gamma^{\mu, \Sigma}$ in the sense of Wasserstein-2 distance as $\varepsilon$ tends to zero,
  that is, $W_2(\nu \circ (X_{1-\varepsilon}^*)^{-1}, \gamma^{\mu, \Sigma}) \rightarrow 0$.
\end{proposition}
\begin{proof}
  Proof follows \cite{dai2023lipschitz}.
\end{proof}

\subsection{Extension of velocity field}
Recall that the velocity field $V(t, x)$ defined in \cref{eq:v-in-s} yields
\[V(t, x) := \frac{\Sigma \nabla \ln \mathcal{Q}_{1-t}r(x)}{t}, \quad r(x) := \frac{p(x)}{\gamma^{\mu, \Sigma}(x)}, \quad t \in (0, 1]\]
and
\[
  \mathcal{Q}_{1-t} r \left( x \right) :=  \int_{\mathbb{R}^{d}} \varphi^{tx + (1-t)\mu, (1-t^2)\Sigma}\left( y \right) r\left( y \right) \mathrm{d}y,
\]
where $\varphi^{tx + (1-t)\mu, (1-t^2)\Sigma}$ is the density function of the $d$-dimensional normal distribution
with mean $tx + (1-t)\mu$ and covariance matrix $(1-t^2)\Sigma$.
For the convenience of subsequent calculation, we introduce the following symbols:
\begin{equation}
  \label{eq:score-def}
  S(t, x) := \nabla \ln q_t(x), \quad q_t(x) := \int_{\mathbb{R}^d} q(t, x| 1, y) p(y)dy, \quad \forall t \in [0, 1]
\end{equation}
where
$q(t, x| 1, y)$ is density function of $N(ty+(1-t)\mu, (1-t^2)\Sigma)$.

Suppose that the target distribution $\nu$ satisfies the third moment condition, we can supplement the definition of velocity field $V$ at time $t=0$,
so that $V$ is well-defined on the interval $[0, 1]$.
\begin{lemma}
  \label{lem:well-posedness}
  Suppose that $\mathbb{E}_\nu[|X|^3] < \infty$, then
  \[
    \lim_{t \rightarrow 0^+} V(t, x) = \lim_{t \rightarrow 0^+}\partial_t S(t, x) = \mathbb{E}_\nu [X] - \mu.\]
\end{lemma}
\newcommand{\offset}{[x - ty - (1-t) \mu]}
\newcommand{\conv}{q(t, x| 1, y)p(y)}
\newcommand{\qbayes}{q(1, y| t, x)}
\begin{proof}
  Since $V(t, x) = \Sigma \nabla \ln \mathcal{Q}_{1-t} r(x) / t = (x - \mu + \Sigma S(t, x))/ t$ for any $t \in (0, 1]$, it yields
  \[ \lim_{t \rightarrow 0^+} V(t, x) = \Sigma \lim_{t \rightarrow 0^+}\partial_t S(t, x)
    = \Sigma \lim_{t \rightarrow 0^+} \left\{ \frac{\nabla[\partial_t q_t(x)]}
    { q_t(x)} - \frac{\partial_t q_t(x) S(t, x)}{q_t (x)}\right\}.\]
  The calculation is similar to \cite{dai2023lipschitz}, we have
  \[ \lim_{t \rightarrow 0^+} V(t, x) = \mathbb{E}_{\nu}[X] - \mu.\]
\end{proof}

We can now extend the flow $(X^*_t)_{t \in [0, 1)}$ to time $t=1$ such that $X^*_1 \sim \gamma^{\mu, \Sigma}$ which solves the IVP
\begin{equation}
  \label{eq:ivp}
  \mathrm{d} X^*_t = -V(1-t, X^*_t)\mathrm{d}t, \quad X^*_0 \sim \nu, \quad t \in [0, 1],
\end{equation}
where the velocity field
\begin{equation*}
  V(1-t, x) = \frac{1}{1-t}[x - \mu + \Sigma S(1-t, x)], \quad \forall t \in [0, 1), \quad V(0, x) = \mathbb{E}_{\nu}[X] - \mu.
\end{equation*}

\subsection{Lipschitz property of velocity field}
\label{subsec:lip-v}
It remains to establish the well-posedness of a flow that solves the IVP \cref{eq:ivp}.
Without loss of generality, we assume that $\mu = \mathbf{0}$ and $\Sigma = s^2 \mathrm{I}_d$ in the following context.
Then by the same techniques used in \cite{dai2023lipschitz}, we can show the Lipschitz property of velocity field.

\begin{theorem}
  \label{thm:grad-v-upper-bd}
  Suppose \cref{cond:cond1,cond:cond2}, \cref{item:non-neg-kappa} or \cref{item:neg-kappa} hold, then
  \begin{equation}
    \nabla V(t, x) \preceq \left( \frac{t D^2}{s^2(1-t^2)^2} - \frac{t}{1-t^2} \right) \mathrm{I}_d, \quad \forall t \in [0, 1).
  \end{equation}
  \begin{equation}
    \nabla V(t, x) \preceq \frac{t(1-\kappa s^2)}{(1-\kappa s^2)t^2 + \kappa s^2} \mathrm{I}_d,
    \  \forall t \in \left[ \sqrt{\frac{\kappa s}{\kappa s -1} \mathds{1}_{\kappa < 0}}, 1 \right].
  \end{equation}
  Suppose that \cref{cond:cond1,cond:cond2,item:conv} hold, then
  \begin{equation}
    \nabla V(t, x) \preceq t \cdot \frac{(\sigma^2-s^2) [s^2+(\sigma^2-s^2)t^2] + R^2 s^2}{[s^2+(\sigma^2-s^2)t^2]^2} \mathrm{I}_d.
  \end{equation}
\end{theorem}

Concerning $\lambda_{\text{max}} (\nabla V(t, x))$, there are the $D^2$-based bound and the $\kappa$-based bound available that can be compared with each other.
One is conditioned on the support assumption and the other one the $\kappa$-semi-log-concave assumption.
We need to decide which one is sharper under certain conditions. Consider the critical case, we get
\begin{equation}
  \label{eq:critical}
  s (D^{-2} - \kappa) = \frac{t^2}{1-t^2}.
\end{equation}
The critical case is $\kappa D^2 = 1$. Note that $\frac{t^2}{1-t^2}$ ranges over $(0, \infty)$ and monotonically increases w.r.t. $t \in (0, 1)$.
Suppose that $\kappa D^2 > 1$, then \cref{eq:critical} has no root over $t \in (0, 1)$.
In this case, $\kappa > 0$ and the $\kappa$-based bound is tighter over $[0, 1)$.
Otherwise, suppose that $\kappa D^2 < 1$, then has a root $t_0 \in (0, 1)$,
and the $D^2$-based bound is tighter over $[0, t_0)$ and the $\kappa$-based bound is tighter over $[t_0, 1)$.

By summarizing the above estimate of the upper bound of $\lambda_{\text{max}} (\nabla V(t, x))$, we obtain the following \cref{thm:jac-v-bd}.

\begin{theorem}
  \label{thm:jac-v-bd}
  Suppose \cref{cond:cond1,cond:cond2}, \cref{item:non-neg-kappa} or \cref{item:neg-kappa} hold.

  If $\kappa D^2 \geq 1$, then
  \[
    \lambda_{\text{max}} (\nabla V(t, x)) \leq \theta_t := \frac{t D^2}{s^2(1-t^2)^2} - \frac{t}{1-t^2}, \quad t \in [0, 1 ].
  \]
  If $\kappa D^2 < 1$, then
  \[
    \lambda_{\text{max}} (\nabla V(t, x)) \leq \theta_t :=
    \begin{cases}
      \frac{tD^2 - s^2t(1-t^2)}{s^2(1-t^2)^2},                 & t \in [0, t_0)  \\
      \frac{t (1-\kappa s^2)}{(1-\kappa s^2)t^2 + \kappa s^2}, & t \in [t_0, 1],
    \end{cases}
  \]
  where $t_0 = \sqrt{\frac{s(1-\kappa D^2)}{(1-\kappa D^2) + D^2}}$ solves \cref{eq:critical}.

  Suppose that \cref{cond:cond1,cond:cond2,item:conv} hold, then
  \[ \lambda_{\text{max}} (\nabla V(t, x)) \leq t \cdot \frac{(\sigma^2-s^2) [s^2+(\sigma^2-s^2)t^2] + R^2 s^2}{[s^2+(\sigma^2-s^2)t^2]^2}. \]
\end{theorem}

In either case, $\lambda_{\text{max}} (\nabla V(t, x))$ is finitely upper bounded by $\theta_t$ over $t \in [0, 1]$,
and so is $\left\| \nabla V(t, x) \right\|_{\mathrm{op}}$, that is to say,
\[
  \left\|\nabla V(t, x)\right\|_{\mathrm{op}} \leq \theta_t, \quad \forall t \in [0, 1].
\]
We now know that the velocity field $-V(x, 1-t)$ is smooth and with bounded derivative on $[0, 1] \times \mathbb{R}^d$.
Therefore the IVP \cref{eq:ivp} has a unique solution and the flow map $x \mapsto X_t^*(x)$ is a diffeomorphism from $\mathbb{R}^d$ to $\mathbb{R}^d$ at any $t \in [0, 1]$. A standard time reversal argument of \cref{eq:ivp} would yield the F\"{o}llmer flow.

\subsection{Lipschitz property of transport maps}
\begin{proof}
  Prove by Gr\"{o}nwall's inequality, we have
  \[
    \lip(X_1(x)) \leq \| \nabla X_1(x) \|_{\mathrm{op}} \leq \exp ( \int_0^1 \theta_s \mathrm{d}s).
  \]
\end{proof}

\section{Convergence of Monte Carlo F\"{o}llmer flow}
\label{app:analysis}
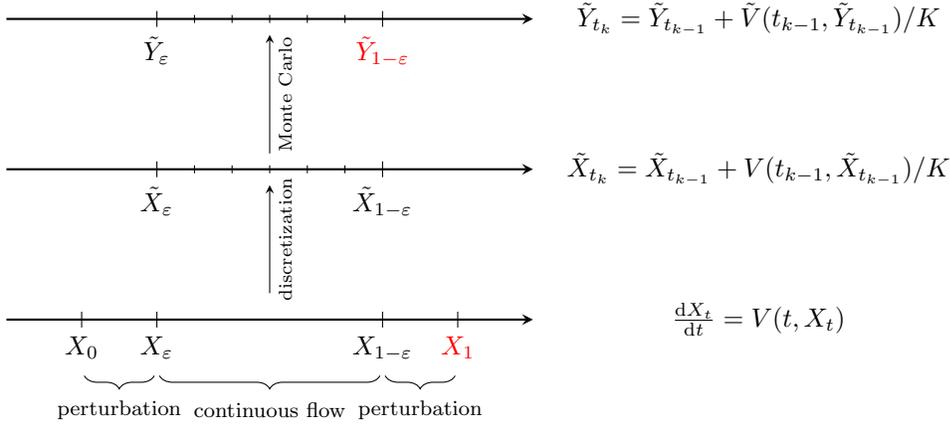
\begin{figure}[htpb]
  \centering
  \label{fig:proof-sketch}
  \begin{tikzpicture}
    \draw[thick,->] (-1,0) -- (6,0);
    \draw (0cm,3pt) -- (0cm,-3pt) node[anchor=north] (X0) {$X_0$};
    \draw (1cm,3pt) -- (1cm,-3pt) node[anchor=north] (Xe) {$X_\varepsilon$};
    \draw (4cm,3pt) -- (4cm,-3pt) node[anchor=north] (X1me) {$X_{1-\varepsilon}$};
    \draw (5cm,3pt) -- (5cm,-3pt) node[anchor=north] (X1) {\color{red} $X_1$};
    \draw [decorate,decoration={brace,amplitude=5pt,mirror,raise=4ex}] (0cm+1pt,-4pt) -- (1cm-1pt,-4pt) node[midway,yshift=-3em]{\footnotesize perturbation};
    \draw [decorate,decoration={brace,amplitude=5pt,mirror,raise=4ex}] (1cm+1pt,-4pt) -- (4cm-1pt,-4pt) node[midway,yshift=-3em]{\footnotesize continuous flow};
    \draw [decorate,decoration={brace,amplitude=5pt,mirror,raise=4ex}] (4cm+1pt,-4pt) -- (5cm-1pt,-4pt) node[midway,yshift=-3em]{\footnotesize perturbation};

    \draw[thick,->] (-1,2) -- (6,2);
    \draw (1cm,2cm+3pt) -- (1cm,2cm-3pt) node[anchor=north] {$\tilde{X}_\varepsilon$};
    \draw (4cm,2cm+3pt) -- (4cm,2cm-3pt) node[anchor=north] {$\tilde{X}_{1-\varepsilon}$};
    \foreach \x in {1.5,2,2.5,3,3.5}
    \draw (\x cm,2cm+1.5pt) -- (\x cm,2cm-1.5pt) node[anchor=north] {};

    \draw[thick,->] (-1,4) -- (6,4);
    \draw (1cm,4cm+3pt) -- (1cm,4cm-3pt) node[anchor=north] {$\tilde{Y}_\varepsilon$};
    \draw (4cm,4cm+3pt) -- (4cm,4cm-3pt) node[anchor=north] {\color{red} $\tilde{Y}_{1-\varepsilon}$};
    \foreach \x in {1.5,2,2.5,3,3.5}
    \draw (\x cm,4cm+1.5pt) -- (\x cm,4cm-1.5pt) node[anchor=north] {};

    \draw[->] (2.5, 4pt+6pt) -- (2.5,2cm-6pt) node [midway, below, sloped] {\scriptsize{discretization}};
    \draw[->] (2.5, 2cm+6pt) -- (2.5,4cm-6pt) node [midway, below, sloped] {\scriptsize{Monte Carlo}};

    \node at(9cm,0cm) {$\frac{\mathrm{d} X_t}{\mathrm{d} t} = V(t, X_t)$};
    \node at(9cm,2cm) {$\tilde{X}_{t_k} = \tilde{X}_{t_{k-1}} + V(t_{k-1}, \tilde{X}_{t_{k-1}})/K$};
    \node at(9cm,4cm) {$\tilde{Y}_{t_k} = \tilde{Y}_{t_{k-1}} + \tilde{V}(t_{k-1}, \tilde{Y}_{t_{k-1}})/K$};
  \end{tikzpicture}
  \caption{Sketch of proof.}
\end{figure}
Without loss of generality, we assume that $\mu = \mathbf{0}$ and $\Sigma = \mathrm{I}_d$ in the following analysis.
\subsection{Preparations}
We first prove that $V$ and $X$ are finite over the unit-time interval.
\subsubsection*{Finite $V$}
The next lemma shows that F\"{o}llmer velocity  is finite.
\begin{lemma}
  \label{lem:v-bd}
  Suppose \cref{cond:cond1,cond:cond2,cond:cond3,cond:lip,cond:bound} hold.
  Let $C$ denote $\frac{\gamma}{\xi}$, then for all $x \in \mathbb{R}^d, t \in [0, 1]$, we have
  \begin{gather*}
    \| V(t, x) \|_2 \leq C.
  \end{gather*}
\end{lemma}
\begin{proof}
  By \cref{eq:v-mc-raw,cond:lip,cond:bound}, we have
  \[ \| V(t, x) \|_2 = \frac{\| \mathcal{Q}_{1-t} \nabla r(x) \|_2}{\mathcal{Q}_{1-t} r(x)} \leq C, \quad \forall t \in (0, 1]. \]
  Since $V(t, x)$ can be continuously extended to $t=0$, the bound holds for $t \in [0, 1]$.
\end{proof}

\subsubsection*{Finite $X$}
The next lemma shows that $\{X_t\}_{t \in [0, 1]}$ is bounded in the sense of expectation.
\begin{lemma}
  \label{lem:x-bd}
  Suppose \cref{cond:cond1,cond:cond2,cond:cond3,cond:lip,cond:bound} hold. Let $C$ denote $\frac{\gamma}{\xi}$, then for all $t \in [0, 1]$,
  \[ \mathbb{E} \|X_t\|_2^2 \leq 2 ( d + C^2) \]
\end{lemma}
\begin{proof}
  By the definition of $X_t$, we have
  \begin{equation*}
    \begin{aligned}
      \|X_t\|_2^2
      \leq & \| X_0\|_2^2 + \left(\int_0^t \| V(X_u, u) \|_2 \mathrm{d} u\right)^2
      + 2 \| X_0\|_2 \int_0^t \| V(u, X_u) \|_2 \mathrm{d} u                                      \\
      \leq & (1+t) \| X_0\|_2^2 + (1+1/t)  \left(\int_0^t \| V(u, X_u) \|_2 \mathrm{d} u\right)^2 \\
      \leq & (1+t) \| X_0\|_2^2 + (t+1)  \int_0^t \| V(u, X_u) \|_2^2 \mathrm{d} u                \\
      \leq & 2 \left( \| X_0\|_2^2 + C^2 \right)
    \end{aligned}
  \end{equation*}
  where the second inequality holds due to $2 a b \leq t a^2 + b^2/t$,
  the third inequality holds by Cauchy-Schwarz inequality,
  and the last inequality holds by \cref{lem:v-bd}.

  Taking expectation, we have
  \[
    \mathbb{E} \|X_t\|_2^2 \leq 2 \left( \mathbb{E} \|X_0\|_2^2 + C^2 \right) \leq  2 (d + C^2).
  \]
\end{proof}

\subsection{Euler's method}
We first derive the upper bound of the first derivative of $V$ with respect to spatial and temporal variable.
Then by Taylor's theorem, we derive the local error and further obtain the accumulated global error.
\subsubsection*{Spatial and temporal derivatives of velocity}
We first derive the Lipschitz constant of the velocity under \cref{cond:lip,cond:bound}.
\begin{lemma}
  \label{lem:der-v}
  Suppose \cref{cond:cond1,cond:cond2,cond:cond3,cond:lip,cond:bound} hold, let $C$ denote $\frac{\gamma}{\xi}$, $L$ denote $C + C^2$.
  Then for all $x, y \in \mathbb{R}^d, t, s \in [\varepsilon, 1-\varepsilon]$, where $\varepsilon \in (0, 1/2)$, we have
  \begin{gather*}
    \| \nabla V(t, x) \|_2 \leq L, \\
    | \partial_t V(t, x) |^2 \leq 4 (C^2+C^4) \left( \| x \|_2^2 + \varepsilon^{-2} d \right).
  \end{gather*}
\end{lemma}
\begin{proof}
  \begin{equation*}
    \begin{aligned}
           & \| V(t, x) - V(t, y) \|_2                                                    \\
      =    & \left\| \frac{\mathcal{Q}_{1-t} \nabla r(x)}{\mathcal{Q}_{1-t} r(x)}
      - \frac{\mathcal{Q}_{1-t} \nabla r(y)}{\mathcal{Q}_{1-t} r(y)} \right\|_2           \\
      \leq & \frac{\| \mathcal{Q}_{1-t} \nabla r(x) - \mathcal{Q}_{1-t} \nabla r(y) \|_2}
      {\mathcal{Q}_{1-t} r(y)} + \|V(t, x)\|_2 \cdot \frac{|\mathcal{Q}_{1-t} r(x) - \mathcal{Q}_{1-t} r(y)|}
      {\mathcal{Q}_{1-t} r(y)}                                                            \\
      \leq & L \| x-y \|_2,
    \end{aligned}
  \end{equation*}
  where the first inequality holds due to $|\frac{a}{b} - \frac{c}{d}| \leq |\frac{a-c}{d}| + |\frac{a}{b} \cdot \frac{b-d}{d}|$,
  and the second inequality holds due to \cref{cond:lip,cond:bound}.
  Similarly, we have
  \begin{equation*}
    \begin{aligned}
           & \| V(t, x) - V(s, x) \|_2^2                                                       \\
      =    & \left\| \frac{\mathcal{Q}_{1-t} \nabla r(x)}{\mathcal{Q}_{1-t} r(x)}
      - \frac{\mathcal{Q}_{1-s} \nabla r(x)}{\mathcal{Q}_{1-s} r(x)} \right\|_2^2              \\
      \leq & 2 \left( \frac{\| \mathcal{Q}_{1-t} \nabla r(x) -
        \mathcal{Q}_{1-s} \nabla r(x) \|_2^2}{|\mathcal{Q}_{1-t} r(x)|^2}
      + \|V(x, s)\|_2^2 \cdot \frac{|\mathcal{Q}_{1-s} r(x)
      - \mathcal{Q}_{1-t} r(x)|^2}{|\mathcal{Q}_{1-t} r(x)|^2} \right)                         \\
      \leq & 2 (C^2 + C^4) \mathbb{E}_{Z \sim \gamma_d} \|
      (t-s) x + (\sqrt{1-t^2} - \sqrt{1-s^2}) Z \|_2^2                                         \\
      \leq & 4 (C^2 + C^4) \left( |t-s|^2 \|x\|_2^2
      + \left| \frac{s^2-t^2}{\sqrt{1-t^2} + \sqrt{1-s^2}} \right|^2 d \right)                 \\
      \leq & 4 (C^2 + C^4) \left| t-s \right|^2 \left( \|x\|_2^2 + \varepsilon^{-2} d \right),
    \end{aligned}
  \end{equation*}
  where the first inequality holds due to $|\frac{a}{b} - \frac{c}{d}|^2 \leq 2 (|\frac{a-c}{d}|^2 + |\frac{a}{b} \cdot \frac{b-d}{d}|^2)$,
  the second inequality holds due to \cref{cond:lip,cond:bound},
  and the last inequality holds due to $\frac{|s^2 - t^2|}{\sqrt{1-t^2} + \sqrt{1-s^2}} \leq \frac{|s-t|}{\sqrt{1-(1-\varepsilon)^2}} \leq \varepsilon^{-1} |s-t| $.
\end{proof}

\begin{theorem}[Truncation error]
  \label{thm:trunc}
  Suppose that \cref{cond:cond1,cond:cond2,cond:cond3,cond:lip,cond:bound} hold. There exists constant $C$, such that
  \begin{equation}
    \label{eq:trunc-error}
    \| X_\varepsilon - X_0 \|_2^2 \leq \varepsilon C^2, \quad \| X_{1} - X_{1-\varepsilon} \|_2^2 \leq \varepsilon C^2.
  \end{equation}
\end{theorem}
\begin{proof}
  Prove by \cref{lem:der-v}, the Lipschitz property of the velocity field .
\end{proof}

\subsubsection*{Global truncation error of Euler's method}
\begin{theorem}[Discretization error]
  \label{thm:euler}
  Suppose that \cref{cond:cond1,cond:cond2,cond:cond3,cond:lip,cond:bound} holds, then
  \[
    \mathcal{W}_2 (\law(X_{1-\varepsilon}), \law(\tilde{X}_{1-\varepsilon})) \leq \mathcal{O} \left( \varepsilon \right) + \mathcal{O} \left( \frac{d}{\varepsilon^2 K^2} \right).
  \]
\end{theorem}
\begin{proof}
  Combining \cref{lem:x-bd,lem:der-v}, we know that for all $t \in [\varepsilon, 1-\varepsilon]$,
  \begin{equation}
    \label{eq:v-t-bd}
    \mathbb{E} \left| \partial_t V(t, X_{t_k}) \right|^2 \leq \mathcal{O} \left( \varepsilon^{-2}d \right).
  \end{equation}

  By vector-valued Taylor's theorem for $X_t$ at interval $[t, t+h]$, we know the remainder is controlled by
  \[
    \| R(t) \|_2^2 := \| X_{t+h} - X_{t} - h V(t, X_{t}) \|_2^2 \leq \frac{h^4}{4}
    \sup_{t \in (t, t+h)} \left| \partial_t V(t, X_{t}) \right|^2, \quad \forall t \in (0, 1-h).
  \]

  Taking expectation and by \cref{eq:v-t-bd}, we have
  \begin{equation}
    \label{eq:taylor-trunc-error}
    \mathbb{E} \| R(t) \|_2^2 \leq \mathcal{O} \left( \varepsilon^{-2} d h^4 \right), \quad \forall t \in (\varepsilon, 1-\varepsilon-h).
  \end{equation}

  The global truncation error is bounded by
    {\small
      \begin{equation*}
        \begin{aligned}
               & \| X_{t_{k+1}} - \tilde{X}_{t_{k+1}} \|_2^2                  \\
          =    & \| X_{t_k} + h V(t_k, X_{t_k}) + R(t_k)
          - \tilde{X}_{t_k} - h V(t_k, \tilde{X}_{t_k}) \|_2^2                \\
          \leq & (1+h) \| X_{t_k} - \tilde{X}_{t_k}
          + h V(t_k, X_{t_k}) - h V(t_k, \tilde{X}_{t_k}) \|_2^2
          + (1+\frac{1}{h}) \| R(t_k) \|_2^2                                  \\
          \leq & (1+h)^2 \| X_{t_k} - \tilde{X}_{t_k} \|_2^2
          + (1+h)(1+\frac{1}{h}) h^2 \| V(t_k, X_{t_k})
          - V(t_k, \tilde{X}_{t_k}) \|_2^2 + (1+\frac{1}{h}) \| R(t_k) \|_2^2 \\
          \leq & \left[1 + (3+4L^2)h\right] \|X_{t_k}
          - \tilde{X}_{t_k} \|_2^2  + (1+\frac{1}{h}) \| R(t_k) \|_2^2.
        \end{aligned}
      \end{equation*}
    }
  The first two inequality holds due to $(a+b)^2 \leq (1+h)a^2 + (1+1/h)b^2$
  and the third inequality holds by the Lipschitz property of $V$ stated in \cref{lem:der-v} and $h \in (0, 1)$.

  Taking expectation and by \cref{eq:taylor-trunc-error}, we have
  \[
    \mathbb{E} \| X_{t_{k+1}} - \tilde{X}_{t_{k+1}} \|_2^2 \leq \left[1 + (3+4L^2)h\right]
    \mathbb{E} \| X_{t_k} - \tilde{X}_{t_k} \|_2^2 + \mathcal{O} \left( \varepsilon^{-2} d h^3 \right).
  \]

  By induction, we have
  \begin{equation*}
    \begin{aligned}
           & \mathbb{E} \| X_{1-\varepsilon} - \tilde{X}_{1-\varepsilon} \|_2^2 = \mathbb{E} \| X_{t_K} - \tilde{X}_{t_K} \|_2^2 \\
      \leq & \left[1 + (3+4L^2)h\right]^{K} \mathbb{E} \| X_{t_0}
      - \tilde{X}_{t_0} \|_2^2 + \frac{\left[1 + (3+4L^2)h\right]^K -1}
      {(3 + 4  L^2)h} \cdot \mathcal{O} \left( \varepsilon^{-2} d h^3 \right)                                                    \\
      \leq & e^{3+4L^2} \mathbb{E} \| X_{\varepsilon}
      - \tilde{X}_{\varepsilon} \|_2^2  + \frac{e^{3+4L^2} - 1}{3 + 4 L^2}
      \cdot \mathcal{O} \left( \varepsilon^{-2} d h^2 \right)                                                                    \\
      \leq & C^2 e^{3+4L^2} \varepsilon+ \frac{e^{3+4L^2} - 1}
      {3+4L^2} \cdot \mathcal{O} \left( \varepsilon^{-2} d h^2 \right)                                                           \\
      \leq & \mathcal{O} \left( \varepsilon \right)
      + \mathcal{O} \left( \varepsilon^{-2} d h^2 \right).
    \end{aligned}
  \end{equation*}
  The third inequality holds due to $\tilde{X}_{\varepsilon} \sim \gamma^{\mu, \Sigma}$
  and the perturbation error stated in \cref{eq:trunc-error}.
  The result reveals a trade-off between the two terms on the choice of $\varepsilon$.
\end{proof}

\subsection{Monte Carlo approximation error}
We first address the error of the Monte Carlo approximation error.
Then we study the stability of the ODE system with respect to velocity field to obtain the overall error introduced by Monte Carlo approximation.
\subsubsection*{MC velocity one-step error}
The following lemma shows that the Monte Carlo approximation of the velocity filed can be precise enough as the number of Monte Carlo simulation $M$ tends to infinity.
\begin{lemma}
  \label{lem:mc-v-error}
  Suppose \cref{cond:lip,cond:bound} hold. Let $C$ denote $\frac{\gamma}{\xi}$, then for all $x \in \mathbb{R}^d$ and $t \in [0, 1]$,
  \[ \mathbb{E} \| V(t, x) - \tilde{V}(t, x) \|_2^2 \leq 8 \left( C^2 + C^4 \right) \frac{d}{M}. \]
\end{lemma}

\begin{proof}
  Denote two independent sets of independent copies of $Z \sim N(\mathbf{0}, \mathrm{I}_d)$, that is, $\mathbf{Z} = \{Z_1, \cdots, Z_M \}$
  and $\mathbf{Z}^\prime = \{Z_1^\prime, \cdots, Z_M^\prime \}$. For notation convenience, we denote
  \begin{align*}
     & d^* =  \mathbb{E}_Z \nabla r (tx + \sqrt{1-t^2}Z),                           &  & e^* = \mathbb{E}_Z r (tx + \sqrt{1-t^2}Z),                            \\
     & d_m = \frac{1}{M} \sum_{i=1}^M \nabla r(tx + \sqrt{1-t^2}Z_i),               &  & e_m = \frac{1}{M} \sum_{i=1}^M r(tx + \sqrt{1-t^2}Z_i),               \\
     & d_m^\prime = \frac{1}{M} \sum_{i=1}^M \nabla r(tx + \sqrt{1-t^2}Z_i^\prime), &  & e_m^\prime = \frac{1}{M} \sum_{i=1}^M r(tx + \sqrt{1-t^2}Z_i^\prime).
  \end{align*}
  Notice that $d^*-d_m = \mathbb{E}[d_m^\prime - d_m | \mathbf{Z}]$, then $\|d^*-d_m\|_2^2 \leq
    \mathbb{E}[\| d_m^\prime - d_m \|_2^2 | \mathbf{Z}]$.
  Then for any $t \in (0, 1]$,
  \begin{align}
         & \mathbb{E} \| d^*-d_m \|_2^2 \leq \mathbb{E}\left[\mathbb{E}[\| d_m^\prime - d_m \|_2^2 |
    \mathbf{Z}]\right] = \mathbb{E} \|d_m^\prime - d_m\|_2^2 \nonumber                               \\
    =    & \frac{1}{M} \mathbb{E}_{z, z^\prime} \| \nabla
    r(tx + \sqrt{1-t^2}z) - \nabla r(tx + \sqrt{1-t^2}z^\prime) \|_2^2 \nonumber                     \\
    \leq & \frac{\gamma^2 (1-t^2)}{M}
    \mathbb{E}_{z, z^\prime} \| z - z^\prime \|_2^2 \nonumber                                        \\
    \leq & \frac{4 \gamma^2 d}{M}, \label{eq:sup-d-and-dm}
  \end{align}
  where the first inequality holds by the Lipschitz property of $\nabla r$.

  Similarly, for any $t \in (0, 1]$, we can also derive
  \begin{align}
      & \mathbb{E} |e^*-e_m|^2 \leq \mathbb{E} |e_m^\prime-e_m|^2 \nonumber \\
    = & \frac{1}{M} \mathbb{E}_{z,z^\prime}\left|
    r (tx + \sqrt{1-t^2}z) - r (tx + \sqrt{1-t^2}z^\prime) \right|^2 \leq \frac{4 \gamma^2 d}{M}, \label{eq:sup-e-and-em}
  \end{align}
  where the inequality holds by the Lipschitz property of $r$.

  Thus by \cref{eq:sup-d-and-dm,eq:sup-e-and-em}, we have,
  \begin{align}
    \| V(t, x) - \tilde{V}(t, x) \|_2^2 =
         & \left\| \frac{d^*}{e^*} - \frac{d_m}{e_m} \right\|_2^2 \nonumber \\
    \leq & 2 \left( \frac{\|d^* - d_m\|_2^2}{|e_m|^2}
    + \frac{\|d^*\|_2^2 |e^*-e_m|^2}{|e^* e_m|^2} \right)       \nonumber   \\
    \leq & 2 \left( \frac{\|d^* - d_m\|_2^2}{\xi^2}
    + \frac{\gamma^2 |e^*-e_m|^2}{\xi^4} \right)                                         \label{eq:de-decomposition},
  \end{align}
  where the first inequality holds due to $|\frac{a}{b} - \frac{c}{d}|^2 \leq 2 (|\frac{a-c}{d}|^2 + |\frac{a}{b} \cdot \frac{b-d}{d}|^2)$,
  and the second inequality holds due to \cref{cond:lip,cond:bound}.

  Combining $\tilde{V}(0, x)= V(0, x) = \mathbb{E}_\nu [X] - \mu$, \cref{eq:sup-d-and-dm,eq:sup-e-and-em,eq:de-decomposition,cond:bound},
  we can conclude that for all $x \in \mathbb{R}^d$ and $t \in [0, 1]$,
  \[ \mathbb{E} \|V(t, x) - \tilde{V}(t, x)\|_2^2 \leq 8 \left( C^2 + C^4 \right) \frac{d}{M}. \]
\end{proof}

\subsubsection*{Monte Carlo approximation error in discrete system}
Then we show the stability of the discrete time ODE flow using the Euler's method.
We consider the difference between the flow $\tilde{X}_{t_k}$ with accurate velocity field and another one $\tilde{Y}_{t_k}$ with Monte Carlo velocity field.
The result is a discrete version of the Alekseev-Gr\"{o}bner formula.
\begin{theorem}[Monte Carlo approximation error]
  \label{thm:discrete-ag-form}
  Suppose that \cref{cond:cond1,cond:cond2,cond:cond3,cond:lip,cond:bound} holds, then
  \[
    \mathcal{W}_2 (\law(\tilde{Y}_{1-\varepsilon}), \law(\tilde{X}_{1-\varepsilon})) \leq \mathcal{O} \left( \frac{d}{KM}\right).
  \]
\end{theorem}
\begin{proof}
  By the definition of $\tilde{X}_{t_k}$ and $\tilde{Y}_{t_k}$, for $k=0, 1, \cdots, K-1$, we have
  \begin{equation*}
    \begin{aligned}
           & \left\| \tilde{Y}_{t_{k+1}} - \tilde{X}_{t_{k+1}} \right\|_2^2   \\
      \leq & (1+h) \left\| \tilde{Y}_{t_{k}} - \tilde{X}_{t_{k}} \right\|_2^2
      + (1 + \frac{1}{h}) h^2 \left\| \tilde{V}(t_k, \tilde{Y}_{t_k})
      - V(t_k, \tilde{X}_{t_k}) \right\|_2^2                                  \\
      \leq & (1+h) \left\| \tilde{Y}_{t_{k}} - \tilde{X}_{t_{k}} \right\|_2^2
      + 2(h + h^2) \left\| \tilde{V}(t_k, \tilde{Y}_{t_k})
      - V(t_k, \tilde{Y}_{t_k}) \right\|_2^2                                  \\
           & + 2(h + h^2) \left\| V(t_k, \tilde{Y}_{t_k})
      - V(t_k, \tilde{X}_{t_k}) \right\|_2^2                                  \\
      \leq & \left[1+(1+4L^2)h\right] \| \tilde{Y}_{t_{k}}
      - \tilde{X}_{t_{k}} \|_2^2 + 2(h + h^2) \|
      \tilde{V}(t_k, \tilde{Y}_{t_k}) - V(t_k, \tilde{Y}_{t_k}) \|_2^2.
    \end{aligned}
  \end{equation*}
  The last inequality holds by the Lipschitz property of $V$ stated in \cref{lem:v-bd} and the fact that $h \in (0, 1)$.

  Taking expectation and by \cref{lem:mc-v-error}, we have
  \[
    \mathbb{E} \left\| \tilde{Y}_{t_{k+1}} - \tilde{X}_{t_{k+1}} \right\|_2^2 \leq
    \left[ 1 + \left( 4L^2 + 1 \right) h \right] \mathbb{E} \left\| \tilde{Y}_{t_{k}} - \tilde{X}_{t_{k}} \right\|_2^2 + 16 (C^2 + C^4)(h + h^2) \frac{d}{M},
  \]

  By induction and the fact that $\tilde{Y}_{1-\varepsilon} = \tilde{X}_{1-\varepsilon}$, we have
  \begin{equation*}
    \begin{aligned}
           & \mathbb{E} \| \tilde{Y}_{1-\varepsilon} - \tilde{X}_{1-\varepsilon} \| = \mathbb{E} \| \tilde{Y}_{t_{K}} - \tilde{X}_{t_{K}} \|_2^2 \\
      \leq & \left[ 1 + \left(1 + 4L^2 \right) h \right]^K
      \mathbb{E} \| \tilde{Y}_{t_0} - \tilde{X}_{t_0} \|_2                                                                                       \\
           & + \frac{\left[ 1 + \left(1 + 4L^2 \right) h \right]^K
      - 1}{\left( 1 + 4L^2 \right)h} \cdot 16 (C^2 + C^4) (h + h^2) \frac{d}{M}                                                                  \\
      \leq & \mathcal{O} \left( \frac{d}{KM}\right).
    \end{aligned}
  \end{equation*}
\end{proof}

\subsection*{Proof of \texorpdfstring{\cref{thm:main}}{main theorem}, Overall error}
\begin{proof}
  We first decompose the error into three parts,
  \[
    \tilde{Y}_{1-\varepsilon} - X_1 = ( \tilde{Y}_{1-\varepsilon} - \tilde{X}_{1-\varepsilon} ) + ( \tilde{X}_{1-\varepsilon} - X_{1-\varepsilon} ) + ( X_{1-\varepsilon} - X_1 ).
  \]
  Then by \cref{thm:trunc,thm:euler,thm:discrete-ag-form}, we obtain that
  \[
    \mathcal{W}_2 (\law(\tilde{Y}_{1-\varepsilon}), \law(X_1)) \leq
    \mathcal{O} \left( \varepsilon \right) + \mathcal{O} \left( \frac{d}{\varepsilon^2 K^2} \right) + \mathcal{O} \left( \frac{d}{KM} \right).
  \]
  By choosing $\varepsilon = \mathcal{O} \left( \sqrt[3]{\frac{d}{K^2}} \right)$, we complete the proof.
\end{proof}

\section{Numerical settings}
\label{app:settings}
\subsection{Examples}
\label{app:example}
We consider the following 11 examples in \cref{table:example}.
\begin{table}[htb]
  \caption{Numerical examples of Gaussian mixture distributions.}
  \label{table:example}
  \resizebox*{\textwidth}{!}{
    \begin{tabular}{ l l l }
      \toprule
      Example & Density                                                                      & Details                                                     \\
      \midrule
      (1)     & $\nu(x) = \frac{1}{4} N(x; -2, 0.25) +  \frac{3}{4} N(x; 2, 0.25)$           &                                                             \\
      (2)     & $\nu(x) = \frac{1}{4} N(x; -4, 0.25) +  \frac{3}{4} N(x; 4, 0.25)$           &                                                             \\
      (3)     & $\nu(x) = \frac{1}{4} N(x; -8, 0.25) +  \frac{3}{4} N(x; 8, 0.25)$           &                                                             \\
      4       & $\nu(x) = \sum_{i=1}^{8} N(x; \alpha_i, 0.03\mathrm{I}_2)$                   & $\alpha_i = 4(\sin(2(i-1)\pi/8), \cos(2(i-1)\pi/8))^\top$   \\
      5       & $\nu(x) = \sum_{i=1}^{16} N(x; \alpha_i, 0.03\mathrm{I}_2)$                  & $\alpha_i = 8(\sin(2(i-1)\pi/16), \cos(2(i-1)\pi/16))^\top$ \\
      6       & $\nu(x) = \sum_{i=1}^{4} \sum_{j=1}^{4} N(x; \alpha_{ij}, 0.03\mathrm{I}_2)$ & $\alpha_{ij} = (2i-5, 2j-5)^\top$                           \\
      7       & $\nu(x) = \sum_{i=1}^{4} \sum_{j=1}^{4} N(x; \alpha_{ij}, 0.03\mathrm{I}_2)$ & $\alpha_{ij} = 2(2i-5, 2j-5)^\top$                          \\
      8       & $\nu(x) = \sum_{i=1}^{5} \sum_{j=1}^{5} N(x; \alpha_{ij}, 0.03\mathrm{I}_2)$ & $\alpha_{ij} = 3(i-3, j-3)^\top$                            \\
      9       & $\nu(x) = \sum_{i=1}^{7} \sum_{j=1}^{7} N(x; \alpha_{ij}, 0.03\mathrm{I}_2)$ & $\alpha_{ij} = 3(i-4, j-4)^\top$                            \\
      10      & $\nu(x) = \sum_{i=1}^{2} \sum_{j=1}^{2} N(x; \alpha_{ij}, \Sigma_{ij})$        & $\alpha_{ij} = (6i-6, 6j-6)^\top, \Sigma_{ij}=[1, (-1)^{i+j+1} 0.9; (-1)^{i+j+1} 0.9, 1]$               \\
      11      & \makecell[l]{$\nu(x) = \frac{1}{5} N(x; -\alpha_d, 0.25\mathrm{I}_d) $                                                                     \\ \qquad \quad $+ \frac{4}{5} N(x; \alpha_d, 0.25\mathrm{I}_d)$} & $\alpha_d = (1, \cdots, 1)^\top \in \mathbb{R}^d$ \\
      \bottomrule
    \end{tabular}
  }
\end{table}

\subsection{Hyperparameters}
For MCMC methods, the number of burn in samples are set to 10,000, the step size is set to 0.2.
Without explicitly specifying, the time grid of F\"{o}llmer flow is uniformly discretized by $K=100$.
For example 4, the preconditioner is set by $\Sigma = 2^2 \mathrm{I}_d$;
For example 5, the preconditioner is set by $\Sigma = 4^2 \mathrm{I}_d$;
For example 7, the preconditioner is set by $\Sigma = 2^2 \mathrm{I}_d$;
For example 8, the preconditioner is set by $\Sigma = {1.7}^2 \mathrm{I}_d$;
For example 9, the preconditioner is set by $\Sigma = {2.1}^2 \mathrm{I}_d$.

\nocite*{}
\bibliographystyle{siamplain}
\bibliography{references}
\end{document}